%% file: tfsm-journal.tex
\documentclass[a4paper,10pt]{article}

\usepackage[T1]{fontenc}
\usepackage[utf8]{inputenc}

	\usepackage{amsmath,amsfonts,amsthm}

\usepackage{newtxtext,newtxmath}%

	\usepackage{graphicx}
	\usepackage{xspace}
	\usepackage{tikz}
	\usetikzlibrary{arrows,automata,shapes}
	\usepackage[small]{subfigure}
	\usepackage{paralist}
	\usepackage{dsfont}
	\usepackage{bbm}
	\usepackage{xfrac}

	\usepackage[noend]{algpseudocode}
	\usepackage{algorithm}
\usepackage{authblk}
	
	\input{macro}

\theoremstyle{plain}
\newtheorem{theorem}{Theorem}
\newtheorem{lemma}{Lemma}
\newtheorem{corollary}{Corollary}

\theoremstyle{definition}
\newtheorem{definition}{Definition}

\begin{document}
	\title{Equivalence Checking and Intersection of Deterministic Timed Finite State Machines}

	\author[1]{Davide Bresolin}
	\author[2]{Khaled El-Fakih}
	\author[3]{Tiziano Villa}
	\author[4]{Nina Yevtushenko}

\date{}                     


\affil[1]{%
		Dipartimento di Matematica,  
		University of Padova, Italy. 
		\texttt{davide.bresolin@unipd.it}}
\affil[2]{%
		American University of Sharjah, 
		Sharjah, United Arab Emirates.
		\texttt{kelfakih@aus.edu}}
\affil[3]{%
		Dipartimento di Informatica,  
		University of Verona, Italy. 
		\texttt{tiziano.villa@univr.it}}
\affil[4]{%
		Ivannikov Institute for System Programming of the Russian Academy of Sciences \& National Research University, Higher School of Economics,
		Moscow, Russia. 
		\texttt{evtushenko@ispras.ru}}
	
	\renewcommand\Affilfont{\small}


	\date{}

	\maketitle

	\begin{abstract}
	There has been a growing interest in defining models of automata 
        enriched with time, such as finite automata extended with clocks
        (timed automata).
	In this paper, we study deterministic timed finite state machines 
        (TFSMs), i.e., finite state machines with a single clock, timed guards 
        and timeouts which transduce timed input words into timed output words.
	We solve the problem of equivalence checking by defining a bisimulation
        from timed FSMs to untimed ones and viceversa.
        Moreover, we apply these bisimulation relations to build the
        intersection of two timed finite state machines by untiming them,
        intersecting them and transforming back to the timed intersection.
	\end{abstract}

\section{Introduction}
\label{sec:intro}

Finite automata (FA) and finite state machines (FSMs) are formal models 
widely used in the practice of engineering and science, e.g.,
in application domains ranging from sequential circuits, communication 
protocols, embedded and reactive systems, to biological modelling. 

Since the 90s, the standard classes of FA have been enriched with 
the introduction of time constraints to represent more accurately 
the behaviour of systems in discrete or continuous time.
Timed automata (TA) are such an example: they are finite automata augmented 
with a number of resettable real-time clocks, whose transitions 
are triggered by predicates involving clock values~\cite{Alur-tcs1994}.

More recently, timed models of FSMs (TFSMs) have been proposed in the literature
by the introduction of time constraints such as timed guards or timeouts.
Timed guards restrict the input/output transitions to happen within given
time intervals. 
The meaning of timeouts is the following: if no input is applied at 
a current state for some timeout period, the timed FSM moves from 
the current state to another state using a timeout function;
e.g., timeouts are common in telecommunication protocols and systems. 

For instance, the timed FSM proposed in~\cite{Gromov2009,ElFakih2013,ElFakih-scp2014} features: one clock variable, time constraints 
to limit the time elapsed at a state, and a clock reset when a transition
is executed.
Instead, the timed FSM proposed in~\cite{Merayo2008,Hierons-jlap2009} 
features: one clock variable, time constraints to limit the time elapsed 
when an output has to be produced after an input has been applied to the FSM, 
a clock reset when an output is produced, and timeouts. 

%
%
  
In~\cite{BresolinEVY14} the following models of deterministic TFSMs with 
a single clock were investigated: 
TFSMs with only timed guards, TFSMs with only timeouts, and TFSMs with 
both timed guards and timeouts.

\begin{figure}[tbp]
   \centering
        \begin{tikzpicture}[xscale=1.5,font=\footnotesize]
                \node (tfsm) at (0,0)   {TFSM with timed guards and timeouts};
                \node (fsm-to)  at (-2,-0.9) {TFSM with timeouts};
                \node (fsm-tg)  at (2,-0.9) {TFSM with timed guards};
                \node (fsm-lf) at (-2,-2.1) {Loop-free TFSM with timeouts};
                \node (fsm-lcro) at (2,-2.1) {TFSM with LCRO timed guards};
                \node (fsm) at (0,-3) {Untimed FSM};
                \draw[->] (tfsm) -- (fsm-to);
                \draw[->] (tfsm) -- (fsm-tg);
                \draw[->] (fsm-to) -- (fsm-lf);
                \draw[->] (fsm-to) -- (fsm-lcro);
                \draw[->] (fsm-tg) -- (fsm-lf);
                \draw[->] (fsm-tg) -- (fsm-lcro);
                \draw[<->] (fsm-lcro) -- (fsm-lf);
                \draw[->] (fsm-lcro) -- (fsm);
                \draw[->] (fsm-lf) -- (fsm);
        \end{tikzpicture}

   \caption{Comparison of TFSM models.}\label{fig:comparison}
\end{figure}
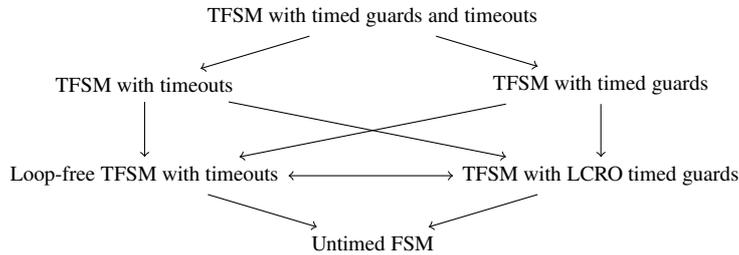

The problem of equivalence checking was solved for all three models, 
their expressive power compared, and subclasses of TFSMs with timeouts
and with timed guards equivalent to each other were characterized (see 
Fig.~\ref{fig:comparison} from~\cite{BresolinEVY14} for a diagram showing 
the expressivity hierarchy of TFSMs with timed guards and timeouts, TFSMs with only timed guards, TFSMs with only timeouts, loop-free TFSMs with timeouts, TFSMs with LCRO - Left Closed Right Open - timed guards, and finally untimed FSMs).
Equivalence checking was obtained by introducing relations of bisimulation that 
define untimed finite state machines whose states include information
on the clock regions, such that the timed behaviours of
two timed FSMs are equivalent if and only if the behaviours of the companion
untimed FSMs are equivalent.
This operation is reminiscent and stronger than the region graph construction
for timed automata~\cite{Alur-tcs1994}. 

Here we work directly with deterministic TFSMs with both timed guards 
and timeouts, since they subsume the previous two models.
For such TFSMs, we give the detailed construction of the untimed FSM 
from a timed FSM (what we get is the \emph{FSM abstraction} of the TFSM),
and then we provide the complete proof that we can describe the
behavior of a TFSM using the corresponding untimed FSM,
i.e., that two deterministic TFSMs are equivalent if and only if their 
timed-abstracted FSMs are equivalent.

Then we study the conditions under which the opposite transformation is possible: we take an untimed deterministic FSM that accepts and produces words from input and output alphabets (both including a special symbol that simulates the passing of time), and we build an equivalent deterministic TFSM with timeouts and timed guards, under the same notion of abstraction of timed words.
This is the key technical result of this paper.

Finally, we apply the previous transformations to perform
the intersection of two deterministic TFSMs, as an example of 
composition operator under which TFSMs are closed.
We prove how the transformation from TFSMs to untimed FSMs of Section~\ref{sec:timedfsms} and the transformation from untimed FSMs to TFSMs of Section~\ref{sec:untimed2timed} can be used to construct the intersection of two TFSMs.

We outline the structure of the paper. 
Sec.~\ref{sec:timedfsms} introduces deterministic timed finite state machines 
with timed guards and timeouts, describes the untiming procedure to obtain 
a finite state machine and proves the bisimulation with the original timed one,
from which an equivalence checking procedure follows. 
This is a revision of the material in ~\cite{BresolinEVY14}, whereas
the following sections are completely new.
Sec.~\ref{sec:untimed2timed} describes the backward transformation 
from untimed FSMs to TFSMs and proves the backward bisimulation relation.
The two results are used in Sec.~\ref{sec:intersection-new} to compute 
the TFSM that is the intersection of two given deterministic TFSMs.
Sec.~\ref{sec:tfsm-ta} relates TFSMs to timed automata, and surveys 
expressiveness and complexity results of various models of timed automata,
with final conclusions drawn in Sec.~\ref{sec:conclusions}.


\section{Models of Timed FSMs (TFSMs)}
\label{sec:timedfsms}

	Let $A$ be a finite alphabet, and let $\bbR^+$ be the set of non-negative reals. A \emph{timed symbol} is a pair $(a,t)$ where $t \in \bbR^+$ is called the \emph{timestamp} of the symbol $a\in A$. 
	A timed word is then defined as a finite sequence $(a_1,t_1)(a_2,t_2)(a_3,t_3)\dots$ of timed symbols where the sequence of timestamps $t_1 \leq t_2 \leq t_3  \leq \dots$ is increasing.
Timestamps represent the \emph{absolute times} at which symbols are received or produced. In the following we will sometime also reason in terms of \emph{relative times}, or \emph{delays}, measured as the difference between the timestamps of two successive symbols. More formally, the delay of a symbol $a_i$ is defined as $t_{i} - t_{i-1}$ when $i > 1$ and as $t_1$ when $i = 1$.

The timed models considered in this paper are initialized input/output machines that operate by reading a \emph{timed input word} $(i_1,t_1)(i_2,t_2)\dots(i_k,t_k)$ defined on some \emph{input alphabet $I$}, and producing a corresponding \emph{timed output word} $(o_1,t_1)$ $(o_2,t_2)\dots(o_k,t_k)$ on some \emph{output alphabet $O$}. 
	The production of outputs is assumed to be instantaneous: the timestamp of the $j$-th output $o_j$ is the same of the $j$-th input $i_j$. Models where there is a delay between reading an input and producing the related output are possible but not considered here.
	Given a timed word $(a_1,t_1)(a_2,t_2)\dots(a_k,t_k)$, 
	$\Untime((a_1,t_1)(a_2,t_2)\dots(a_k,t_k)) = a_1 a_2 \dots a_k$ denotes
	the word obtained when deleting the timestamps.

A timed possibly non-deterministic and partial FSM (TFSM) is an FSM 
augmented with a clock. The clock is a real number that measures the time delay
at a state, and its value is reset to zero when a transition is executed.
In this section we introduce the TFSM model with both timed guards and timeouts defined in~\cite{BresolinEVY14}.
Such a model  subsumes the TFSM model with timed guards only given in~\cite{ElFakih2013,Gromov2009} and the TFSM model with timeouts only given in~\cite{Merayo2008,Zhigulin2011}. In addition, we establish a very precise connection between timed and untimed FSMs, showing that it is possible to describe the behavior of a TFSM using a standard FSM
that is called the \emph{FSM abstraction} of the TFSM. 

A {\em timed guard} defines the time interval when a transition can be executed.
Intuitively, a TFSM in the present state $s$ accepts an input $i$ at a time 
$t$ only if $t$ satisfies the timed guard of some transition labelled with input symbol $i$.
The transition defines the output $o$ to be produced and the next 
state $s'$. 
A {\em timeout} instead defines for how long the TFSM can wait for an input in the present state before spontaneously moving to another state. Each state of the machine has a timeout (possibly $\infty$) and all outgoing transitions of the state have timed guards with upper bounds less than the state timeout. 
The clock is reset to 0 every time the TFSM activates a transition or a timeout expires. 
 
\begin{definition}[Timed FSM]\label{def:tfsm-all}
A timed FSM $M$ is a finite state machine augmented with timed guards and timeouts. Formally, a timed FSM (TFSM) is a 6-tuple $(S, I, O, \lambda_S, s_0,$ $\Delta_S)$
where $S$, $I$, and $O$ are finite disjoint non-empty sets of states, inputs and outputs, respectively, $s_0$ is the initial state,
$\lambda_S \subseteq  S \times \left(I \times \Pi \right) \times  O \times  S$ 
is a transition relation where $\Pi$ is the set of input timed guards, and
$\Delta_S: S \rightarrow S \times \left( N \cup \left\{ \infty \right\} \right)$
is a \textit{timeout function}.
Each guard in $\Pi$ is an interval $g = \langle t_{min}, t_{max} \rangle$ where $t_{min}$ is a nonnegative integer, while $t_{max}$ is either a nonnegative integer or $\infty$, $t_{min} \leq t_{max}$, and  $\langle \in \big\{ ( , [ \big\}$ while $\rangle \in \big\{ ), ] \big\}$. 
\end{definition}

The \emph{timed state} of a TFSM is a pair $(s,x)$ such that $s \in S$ is a state of $M$ and $x \in \bbR^+$ is the current value of the clock, with the additional constraint that $x < \Delta_S(s)_{\downarrow \bbN}$ (the value of the clock cannot exceed the timeout). If no input is applied at a current state $s$ before the timeout 
$\Delta_S\left(s\right)_{\downarrow \bbN}$ expires, then the TFSM will move to anther state 
$\Delta_S\left(s\right)_{\downarrow S}$ as prescribed by the timeout function. 
If $\Delta_S\left(s\right)_{\downarrow \bbN} = \infty$, then the TFSM can stay at state $s$ infinitely long waiting for an input. An input/output transition can be triggered only if the value of the clock is inside the guard $\langle t_{min}, t_{max}\rangle$ labeling the transition.  
\emph{Transitions} between timed states can be of two types:

\begin{compactitem}
	\item \emph{timed transitions} of the form $(s,x) \trans{t} (s',x')$ where $t \in \bbR^+$, representing the fact that a delay of $t$ time units has elapsed without receiving any input. The relation $\trans{t}$ is the smallest relation closed under the following properties:
		\begin{compactitem}
			\item for every timed state $(s,x)$ and delay $t \geq 0$, if $x + t < \Delta_S(s)_{\downarrow \bbN}$, then $(s,x) \trans{t} (s, x+t)$;
			\item for every timed state $(s,x)$ and delay $t \geq 0$, if $x + t = \Delta_S(s)_{\downarrow \bbN}$, then $(s,x) \trans{t} (s', 0)$ with $s' = \Delta_S(s)_{\downarrow S}$;
			\item if $(s,x) \trans{t_1} (s',x')$ and $(s',x') \trans{t_2} (s'',x'')$ then $(s,x) \trans{t_1+t_2} (s'',x'')$.
			\end{compactitem}

	\item \emph{input/output transitions} of the form $(s,x) \trans{i,o} (s',0)$, representing reception of the input symbol $i \in I$, production of the output $o \in O$ and reset of the clock. An input/output transition can be activated only if there exists $(s,i,\langle t_{min}, t_{max}\rangle,o,s') \in \lambda_S$ such that $x \in \langle t_{min}, t_{max}\rangle$.
\end{compactitem}

\noindent A \emph{timed run} of a TFSM $M$ interleaves timed transitions with input/output transitions.
	Given a timed input word $v = (i_1,t_1)(i_2,t_2)\dots$ $(i_k,t_k)$, a timed run of $M$ over $v$ is a finite sequence $\rho = (s_0,0) \trans{t_1} (s_0',x_0) \trans{i_1,o_1} (s_1,0) \trans{t_2-t_1} (s_1',x_1) \trans{i_2,o_2} (s_2,0) \trans{t_3-t_2} \dots  \trans{i_k,o_k} (s_k,0)$ such that $s_0$ is the initial state of $M$, and for every $j \geq 0$ $(s_j,0) \trans{t_{j+1}-t_j} (s_j',x_j) \trans{i_{j+1},o_{j+1}} (s_{j+1},0)$ is a valid sequence of transitions of $M$. The timed run $\rho$ is said to \emph{accept} the timed input word $v= (i_1,t_1)(i_2,t_2)\dots(i_k,t_k)$ and to \emph{produce} the timed output word $u = (o_1,t_1)(o_2,t_2)\dots(o_k,t_k)$. 	The behavior of $M$ is defined in terms of the input/output words accepted and produced by the machine.

%

The usual definitions for FSMs of deterministic and non-deterministic, submachine, etc., can be extended to the timed FSM model considered here. In particular, a TFSM is {\em complete} if for each state $s$, input $i$ and value of the clock $x$ there exists at least one transition $(s,x) \trans{i,o} (s',0)$,  otherwise the machine is {\em partial}. A TFSM is \emph{deterministic} if for each state $s$, input $i$ and value of the clock $x$  there exists at most one input/output transition, otherwise is {\em non-deterministic}.

 	For the sake of simplicity, from now on we consider only \emph{deterministic} machines (possibly partial), leaving the treatment of non-deterministic TFSMs to future work.

	\begin{definition}\label{def:TGTO_behavior}
	The behavior of a deterministic TFSM $M$ is a partial mapping $B_M: (I\times \bbR)^* \mapsto (O\times \bbR)^*$ that associates every input word $w = (i_1,t_1)(i_2,t_2)\dots(i_k,t_k)$ accepted by $M$ with the unique output word $B_M(w) = (o_1,t_1)(o_2,t_2)\dots(o_k,t_k)$ produced by $M$ under input $w$, if it exists. When $M$ is an untimed FSM the behavior is defined as a partial mapping $B_M: I^* \mapsto O^*$.

Two machines $M$ and $M'$ with the same input and output alphabets are \emph{equivalent} if and only if they have same behavior, i.e, $B_M = B_{M'}$.
	\end{definition}

	So for a partial and deterministic TFSM $M$, we have that for every input word $w$, $B_M(w)$ is either not defined or a singleton set. 
	Moreover, we can consider the transition relation of the machine as a partial function $\lambda_S: S \times I \times \bbR^+ \mapsto S \times O$ that takes as input the current state $s$, the delay $t$ and the input symbol  $i$ and produces the (unique) next state and output symbol $\lambda_S(s,t,i) = (s',o)$ such that $(s,0)\trans{t}(s',t')\trans{i,o}(s'',0)$. With a slight abuse of the notation, we can extend it to a partial function $\lambda_S: S \times (I \times \bbR^+)^* \mapsto S \times O^*$ that takes as inputs the initial state $s$ and a timed word $w$, and returns the state reached by the machine after reading $w$ and the generated output word. We will use $s \trans{w,u} s'$ as a shorthand for $\lambda_S(s,w) = (s',u)$.


\paragraph{Abstracting TFSMs with timeouts and timed guards.}

In this section we show how to build an abstract untimed FSM that describes the behaviour of a TFSM with guards. To do this we define an appropriate notion of abstraction of a timed word into an untimed word and a notion of bisimulation to compare a TFSM with guards with untimed FSM. From the properties of the bisimulation relation, we conclude that the behaviour of the abstract untimed FSM is the abstraction of the behaviour of the TFSM.

For every $N \geq 0$, we define $\bbI_N$ as the set of intervals 
$	\bbI_N = \{[n,n] \mid n \leq N\} \cup \{(n,n+1) \mid 0 \leq n < N\} \cup \{(N,\infty)\}.$
Given a TFSM $M$, we define $\max(M)$ as the maximum between the greatest timeout value of the function $\Delta_S$ (different from $\infty$) and the greatest integer constant (different from $\infty$) appearing in the guards of $\lambda_S$. The set $\bbI_N$ defines a discretization of the clock values of TFSMs. The following lemma proves that such a discretization is correct, namely, that a TFSM cannot distinguish between two timed states where the discrete state is the same and the values of the clocks are in the same interval of $\bbI_N$.

\begin{lemma}\label{lem:discretization}
Let $M= (S, I, O, \lambda_S, s_0, \Delta_S)$ be a deterministic TFSM, $N = \max(M)$, and let $(s,x)$ and $(s,x')$ be two timed states of $M$ such that $x, x' \in \langle n,n'\rangle$ for some interval $\langle n,n'\rangle \in \bbI_N$. Then $\lambda_S(s, x, i)  = \lambda_S(s, x', i)$ for every input symbol $i \in I$.
\end{lemma}

\begin{proof}
Suppose by contradiction that there exist two timed states $(s,x)$ and $(s,x')$ such that $x, x' \in \langle n,n'\rangle$ for some $\langle n,n'\rangle \in \bbI_N$ and $\lambda_S(s, x, i) \neq \lambda_S(s, x', i)$. 
Since $x \neq x'$ we have that the interval $\langle n,n'\rangle$ must be an open interval of the form $(n,n+1)$ (it cannot be a point interval $[n,n]$) with $n = \lfloor x \rfloor = \lfloor x' \rfloor$ and $n + 1 = \lceil x \rceil = \lceil x' \rceil$. Suppose, without loss of generality, that $\lambda_S(s, x, i)$ is defined and equal to $(s', o)$. By the definition of TFSM we have that there exists a transition $(s, i, \langle t_{min}, t_{max}\rangle, o, s') \in \lambda_S$ such that $x \in  \langle t_{min}, t_{max}\rangle$. Since $t_{min}, t_{max}$ are nonnegative integers (or $\infty$), it is easy to see that $(n,n+1) \subseteq \langle t_{min}, t_{max}\rangle$. Hence, $x' \in \langle t_{min}, t_{max}\rangle$ and thus $\lambda_S(s, x', i) = (s', o) = \lambda_S(s, x, i)$, in contradiction with the hypothesis that $\lambda_S(s, x, i) \neq \lambda_S(s, x', i)$.
\end{proof}

We can exploit the discretization given by $\bbI_N$ to build the abstract FSM as follows.
States of the abstract FSM will be pairs $(s,\langle n,n'\rangle)$ where $s$ is a state of $M$ and $\langle n,n'\rangle$ is either a point-interval $[n,n]$ or an open interval $(n,n+1)$ from the set 
$\mathbb{I}_N$ defined above, where $N = \max(M)$.  Transitions can be either standard input/output transitions labelled with pairs from $I \times O$ or ``time elapsing'' transitions labelled with the special pair $(\Half,\Half)$, which intuitively represents a time delay $0 < t^* < 1$ without inputs. 

\begin{definition}\label{def:abstract-tfsm-all}
Given a TFSM with timeouts and timed guards $M = (S, I, O, \lambda_S, s_0, \Delta_S)$, let $N = \max(M)$. We define the \emph{$\Half$-abstract FSM} $A_M = (S\times \mathbb{I}_N, I \cup \{\Half\}, O\cup \{\Half\}, \lambda_A, (s_0,[0,0]))$ as the untimed FSM such that:
\begin{compactitem}
	\item $(s,[n,n]) \trans{\Half,\Half} (s,(n,n+1))$ if and only if $n+1 \leq \Delta_S(s)_{\downarrow\bbN}$;
	\item $(s,(n,n+1)) \trans{\Half,\Half} (s,[n+1,n+1])$ if and only if $n+1 < \Delta_S(s)_{\downarrow\bbN}$;
	\item $(s,(n,n+1)) \trans{\Half,\Half} (s',[0,0])$ if and only if $\Delta_S(s) = (s',n+1)$;
	\item $(s,[N,N]) \trans{\Half,\Half} (s,(N,\infty))$ and $(s,(N,\infty)) \trans{\Half,\Half} (s,(N,\infty))$ if and only if \linebreak $\Delta_S(s)_{\downarrow\bbN} = \infty$;
	\item $(s,\langle n,n'\rangle) \trans{i,o} (s',[0,0])$ if and only if there exists $(s,i,\langle t,t'\rangle,o,s') \in \lambda_S$ such that $\langle n,n'\rangle \subseteq \langle t,t'\rangle$.
\end{compactitem}
\end{definition}

\begin{figure}[tbp]
\centering
	\subfigure[TFSM with timeouts and timed guards $M$]{
	\begin{tikzpicture}[font=\footnotesize,yscale=0.8,xscale=0.65]
\node[draw,circle,minimum size=20pt] (q0)	 at (0,0)	{$s_0$};
\node[draw,circle,minimum size=20pt]	(q1) at (3.5,0)	{$s_1$};
		
\draw[->] (-0.5,0.5) -- (q0);
\path[->] (q0) edge[loop above] node        {$[0,1):i/o_1$} ();
\path[->] (q0) edge[bend left]  node[above] {$t = 1$} (q1);
\path[->] (q1) edge[bend left]  node[below] {$(1,\infty):i/o_1$} (q0);
\path[->] (q1) edge[loop above] node        {$[0,1]:i/o_2$} ();
\path (-3,-2)--(7,2);
	\end{tikzpicture}\label{fig:tfsm-all-ex}}

	\subfigure[$\Half$-Abstract untimed FSM $A_M$]{
\begin{tikzpicture}[yscale=0.75,xscale=1,font=\footnotesize]
		\node[draw,circle,minimum size=20pt,text width=4ex,text centered] (q00)	 at (-4,2.25)	{$s_0$ \\ $[0,0]$};
		\node[draw,circle,minimum size=20pt,text width=4ex,text centered] (q005)	 at (-2,2.25)	{$s_0$ \\ $(0,1)$};
		\node[draw,circle,minimum size=20pt,text width=4ex,text centered]	(q10) at (0,2.25)	{$s_1$ \\ $[0,0]$};
		\node[draw,circle,minimum size=20pt,text width=4ex,text centered]	(q105) at (1,0)	{$s_1$ \\ $(0,1)$};
		\node[draw,circle,minimum size=20pt,text width=4ex,text centered]	(q11) at (-1,0)	{$s_1$ \\ $[1,1]$};
		\node[draw,circle,minimum size=20pt,text width=4ex,text centered]	(q115) at (-3,0)	{$s_1$ \\ $(1,\infty)$};
		
		\draw[->] (-4.5,3.25) -- (q00);
		\path[->] (q00) edge[loop above] node        {$i/o_1$} ();
		\path[->] (q00) edge node[below]  {$\Half/\Half$} (q005);
		\path[->] (q005) edge[bend right] node[above]        {$i/o_1$} (q00);
		\path[->] (q005) edge node[above] {$\Half/\Half$} (q10);
		\path[->] (q10) edge[loop above] node  {$i/o_2$} ();
		\path[->] (q10) edge node[left] {$\Half/\Half$} (q105);
		\path[->] (q105) edge[bend right] node[right]  {$i/o_2$} (q10);
		\path[->] (q105) edge node[below] {$\Half/\Half$} (q11);
		\path[->] (q11) edge[bend left] node[near start,left]  {$i/o_2$} (q10);
		\path[->] (q11) edge node[below] {$\Half/\Half$} (q115);
		\path[->] (q115) edge[bend left] node[left] {$i/o_1$} (q00);
		\path[->] (q115) edge[loop left] node {$\Half/\Half$} (q115);
			\end{tikzpicture}\label{fig:tfsm-all-ab}}

	\caption{{$\Half$-abstraction} of TFSM with timeout and timed guards.}
  \label{fig:tfsm-all-abstraction}
\end{figure}
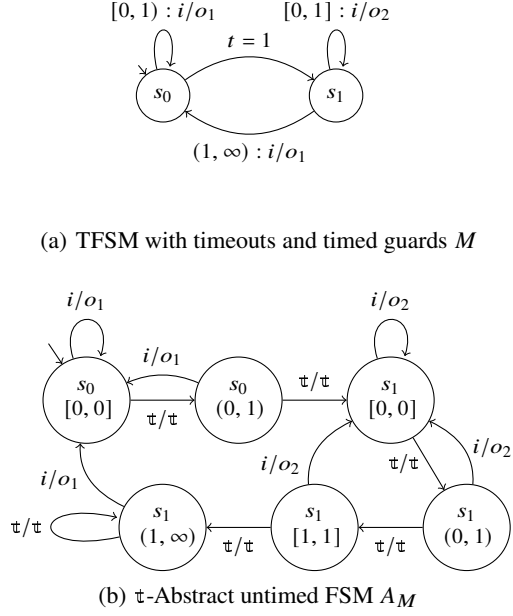

Figure~\ref{fig:tfsm-all-abstraction} shows an example of a TFSM with timeouts and its $\Half$-abstraction. In this case the untimed abstraction accepts untimed input words on $I \cup \{\Half\}$. The delay is implicitly represented by sequences of the special input symbol $\Half$ interleaving the occurrences of the real input symbols from $I$. The representation of delays in the abstraction is quite involved:
\begin{compactitem}
	\item an \emph{even number} $2n$ of $\Half$ symbols represents a delay of \emph{exactly} $n$ time units;
	\item an \emph{odd number} $2n + 1$ of $\Half$ symbols represents a delay $t$ included in the open interval $(n,n+1)$.
\end{compactitem}
The notion of abstraction of a timed word captures the above intuition. 

\begin{definition}\label{def:Half-abstract-timed-word}
Let $\Half(t)$ be a function mapping a delay $t \in \bbR$ to a sequence of $\Half$ as follows:
$\Half(t) = 	\Half^{2t}$ if $\floor{t} = t$,  $\Half(t) = \Half^{2\floor{t} + 1}$ otherwise.
Given a finite alphabet $A$ and a finite timed word $v = (a_1,t_1)$ $(a_2,t_2)(a_3,t_3)\dots(a_m,t_m)$, we define its \emph{$\Half$-abstraction} as the finite word  $\Half(v) = \Half(t_1)a_1 \Half(t_2-t_1) \dots \Half(t_j-t_{j-1}) a_j \Half(t_{j+1}-t_j)\dots \Half(t_{m-1}-t_m) a_m$.
\end{definition}

$\Half$-bisimulation connects \emph{timed states $(s,x)$} of a timed FSM with states of an untimed FSM.
Conditions 1.\ and 2.\ formalize the connection between timed transitions and the special symbol $\Half$. Conditions 3.\ and 4.\ formalize the connection between actual input/output transition in the two machines. 

\begin{definition}\label{def:Half-bisim}
Given a TFSM with timed guards and timeouts  $T = (S, I, O, \lambda_S, s_0, \Delta_S)$ and an untimed FSM $U = (R, I \cup \{\Half\}, O \cup \{\Half\}, \lambda_R, r_0)$, a \emph{$\Half$-bisimulation} is a relation $\sim \subseteq (S\times \bbR^+) \times R$ that respects the following conditions for every pair of states $(s,x) \in S\times \bbR^+$ and $r \in R$ such that $(s,x) \sim r$:
\begin{compactenum}
	\item if $(s,x)\trans{t}(s',x')$ with $0 < t < 1$ and either $x \in \mathbb{N}$ or $x + t \in \mathbb{N}$ then there exists $r'\in R$ such that $r \trans{\Half,\Half} r'$ and $(s',x') \sim r'$;
	\item if $r \trans{\Half,\Half} r'$ then for every $0 < t < 1$ such that either $x \in \mathbb{N}$ or $x + t \in \mathbb{N}$ there exists $(s',x')\in S\times \bbR^+$ such that $(s,x)\trans{t}(s',x')$ and $(s',x') \sim r'$;
\item if $(s,x)\trans{i,o}(s',0)$ then there exists $r'\in R$ such that $r \trans{i,o} r'$ and $(s',0) \sim r'$;	
	\item if $r \trans{i,o} r'$ then there exists $(s',0)\in S\times \bbR^+$ such that $(s,x)\trans{i,o}(s',0)$ and $(s',0) \sim r'$.
\end{compactenum}
$T$ and $U$ are  \emph{$\Half$-bisimilar} if there exists a $\Half$-bisimulation $\sim \subseteq S \times R$ such that $(s_0,0) \sim r_0$.
\end{definition}

To prove that $\Half$-bisimilar machines have the same behavior we need to introduce the following technical result, connecting timed transitions with the special symbol $\Half$.

\begin{lemma}\label{lem:Half-bisimilar-time}
Given a TFSM with timed guards and timeouts $T = (S, I, O, \lambda_S, s_0, \Delta_S)$ and an untimed FSM $U = (R, I\cup\{\Half\}, O\cup\{\Half\}, \lambda_R, r_0)$, every $\Half$-bisimulation relation $\sim \subset (S\times\bbR^+)\times R$ respects the following properties for every $(s,0) \sim r$ and $t > 0$:
\begin{compactenum}[\it (i)]
	\item if $(s,0) \trans{t} (s',x')$ then there exists $r'$ such that $(s',x') \sim r'$ and $r \trans{\Half(t),\Half(t)} r'$;
	\item if $r \trans{\Half(t),\Half(t)} r'$  then there exists $(s',x') \sim r'$ such that $(s,0) \trans{t} (s',x')$.
\end{compactenum}
\end{lemma}

\begin{proof}
The proof is by induction on the number of symbols $n$ in $\Half(t)$.
For the basis of the induction, suppose $n = 1$ and let $(s,0) \sim r$: by the definition of $\Half(t)$, we have that $0 < t < 1$. The two properties are a direct consequence of the definition of $\Half$-bisimulation. By condition 1 of Definition~\ref{def:Half-bisim}, we have that for every $0 < t < 1$, $(s,0) \trans{t} (s',x')$ implies that there exists $r'$ such that $(s',x') \sim r'$ and $r \trans{\Half,\Half} r'$. By condition 2 of Definition~\ref{def:Half-bisim}, we have that for every $0 < t < 1$, $r \trans{\Half,\Half} r'$ implies that there exists $(s',x') \sim r'$ such that $(s,0) \trans{t} (s',x')$.

For the inductive case, suppose that $n \geq 1$ and that the Lemma holds for $n - 1$. Now, let $(s,0) \trans{t} (s',x')$. Two cases may arise: either $\floor{t} = t$ or $\floor{t} > t$. In the former case, consider the timed state $(s'',x'')$ such that $(s,0) \trans{t-0.5} (s'',x'') \trans{0.5} (s',x')$.\footnote{Here $0.5$ is an arbitrary value chosen for the sake of simplicity. Indeed, the argument holds for every delay $0 < t^* < 1$.} Since the number of symbols in $\Half(t-0.5)$ is exactly $n-1$, by inductive hypothesis we have that there exists $r''$ such that $(s'',x'') \sim r''$ and $r \trans{\Half,\Half^{n-1}} r''$. By condition 1 of Definition~\ref{def:Half-bisim}, we have that there exists $r'$ such that $(s',x') \sim r'$ and $r'' \trans{\Half,\Half} r'$ and thus that $r \trans{\Half(t),\Half(t)} r'$. To prove property \textit{(ii)}, suppose $r \trans{\Half(t),\Half(t)} r'$ and consider the state $r''$ such that $r \trans{\Half,\Half^{n-1}} r''\trans{\Half,\Half} r'$. By inductive hypothesis we have that there exists $(s'',x'') \sim r''$ such that $(s,0) \trans{t-0.5} (s'',x'')$. By condition 2 of Definition~\ref{def:Half-bisim} it is possible to find a state $(s',x')$ such that $(s',x') \sim r'$ and $(s'',x'') \trans{0.5} (s',x')$. This shows that $(s,0) \trans{t} (s',x')$. 
When $\floor{t} > t$, we can consider the timed state $(s'',x'')$ such that $(s,0) \trans{\floor{t}} (s'',x'') \trans{t - \floor{t}} (s',x')$. Since the number of symbols in $\Half(\floor{t})$ is exactly $n-1$, by an argument similar to the above we can prove that both properties \emph{(i)} and \emph{(ii)} hold also in this case, concluding the proof. 
\end{proof}

The following lemma proves that $\Half$-bisimilar machines have the same behavior.

\begin{lemma}\label{lem:Half-bisimilar-then-equiv}
Given a TFSM with timeouts $T = (S, I, O, \lambda_S, s_0, \Delta_S)$ and an untimed FSM $U = (R, I\cup\{\Half\}, O\cup\{\Half\}, \lambda_R, r_0)$, if there exists a $\Half$-bisimulation $\sim$ such that $(s_0,0) \sim r_0$ then for every timed input word $v = (i_1,t_1)\dots(i_m,t_m)$ we have that $\Half(B_T(v)) = B_U(\Half(v))$.
\end{lemma}

\begin{proof}
We prove the lemma by showing that the following claim holds:

\begin{quote}
\em for every pair of states $s \in S$ and $r \in R$ such that $(s,0) \sim r$ and timed word $v$, $\lambda_S(s,v) = (s',w)$ if and only if $\lambda_R(r,\Half(v)) = (r',\Half(w))$ with $(s',0) \sim r'$.
\end{quote}

We prove the claim by induction on the length $m$ of the input word. Suppose $m = 1$, $v = (i_1,t_1)$ and $w = (o_1,t_1)$. We have to show that $\lambda_S(s,(i_1,t_1)) = (s',(o_1,t_1))$ if and only if $\lambda_R(r,\Half(i_1,t_1)) = (r',\Half(o_1,t_1))$ for some $r'$ such that $(s',0) \sim r'$.

To prove the direct implication, suppose $\lambda_S(s,v) = (s',w)$. By the definition of TFSM we have that $\lambda_S(s, (i_1,t_1)) = (s_1,(o_1,t_1))$ if and only if there exists a timed state $(s',x')$ such that $(s,0) \trans{t_1} (s',x') \trans{i_1,o_1} (s_1,0)$. We distinguish between two cases depending on the value of $t_1$. 

\begin{compactitem}
	\item If $t_1 = 0$, by condition 3. of the definition of $\Half$-bisimulation (since $(s',x') \trans{i_1,o_1} (s_1,0)$), there exists $r_1 \in R$ such that $r \trans{i_1,o_1} r_1$. Hence, $\lambda_R(r,\Half(i_1,t_1)) = (r_1,\Half(o_1,t_1))$.
	
	\item If $t_1 > 0$, by Lemma~\ref{lem:Half-bisimilar-time} (i), there exists $r'$ such that $r \trans{\Half(t_1),\Half(t_1)} r'$ and $(s',x') \sim r'$. By condition 3. of the definition of $\Half$-bisimulation (since $(s',x') \trans{i_1,o_1} (s_1,0)$), we have that it is possible to find a state $r_1 \in R$ such that $r' \trans{i_1,o_1} r_1$. This implies that under input $\Half(t_1) i_1 = \Half(i_1,t_1)$ the FSM $U$ produces the output word $\Half(t_1) o_1 = \Half(o_1,t_1)$, and thus we can conclude that $\lambda_R(r,\Half(i_1,t_1)) = (r_1,\Half(o_1,t_1))$.
\end{compactitem}

To prove the converse implication, suppose $\lambda_R(r,\Half(i_1,t_1)) = (r_1,\Half(o_1,t_1))$. We distinguish between two cases depending on the value of $t_1$.
\begin{compactitem}
	\item If $t_1 = 0$, then by the assumption $\lambda_R(r,\Half(i_1,t_1)) = (r_1,\Half(o_1,t_1))$ there exists $r' \in R$ such that $r' \trans{i_1,o_1} r_1$, and so by condition 4. of the definition of $\Half$-bisimulation, there exists $(s_1,0) \in S \times \bbR$ such that $(s,t_1) \trans{i_1,o_1} (s_1,0)$. Hence, $\lambda_S(s, (i_1,t_1)) = (s_1,(o_1,t_1))$.	
	
	\item If $t_1 > 0$, then by the assumption $\lambda_R(r,\Half(i_1,t_1)) = (r_1,\Half(o_1,t_1))$ there exists $r' \in R$ such that $r \trans{\Half(t_1),\Half(t_1)} r' \trans{i_1,o_1} r_1$. By Lemma~\ref{lem:Half-bisimilar-time} (ii), there exists $(s',x')\in S \times \bbR$ such that $(s,0) \trans{t_1} (s',x')$ and $(s',x') \sim r'$. By condition 4. of the definition of $\Half$-bisimulation, we have that there exists a timed state $(s_1,0)$ such that $(s',x') \trans{i_1,o_1} (s_1,0)$. This implies that under input $(i_1,t_1)$ the TFSM $T$ produces the timed output word $(o_1,t_1)$, and thus we can conclude that $\lambda_S(s, (i_1,t_1)) = (s_1,(o_1,t_1))$.	
\end{compactitem}

\noindent Since our machines may be partial, we have that $\lambda_S(s,(i_1,t_1))$ and $\lambda_R(r,\Half(v))$ are not necessarily defined. However, the above argument also shows that $\lambda_S(s,(i_1,t_1))$ is defined if and only if $\lambda_R(r,\Half(v))$ is defined.

\medskip

\sloppy
To prove the inductive case, suppose $m > 1$, $v = (i_1,t_1)\ldots(i_m,t_m)$ and $w = (i_1,t_1)\ldots$ $(i_m,t_m)$. Now, let $v'=(i_1,t_1)\ldots(i_{m-1},t_{m-1})$ and $w' = (o_1,t_1)\ldots(o_{m-1},t_{m-1})$. By inductive hypothesis, we have that $\lambda_S(s,v') = (s_{m-1},w')$ if and only if $\lambda_R(r,\Half(v')) = (r_{m-1},\Half(w'))$ for some $(s_{m-1},0) \sim r_{m-1}$, and that $\lambda_S(s_{m-1},(i_m,t_m-t_{m-1})) = (s_m,(o_m,t_m-t_{m-1}))$ if and only if $\lambda_R(r_{m-1},\Half(i_m,t_m-t_{m-1})) = (r_m,\Half(o_m,t_m-t_{m-1}))$ for some $(s_m,0) \sim r_m$. This implies that $\lambda_S(s,v'(i_m,t_m)) = (s_m,w'(o_m,t_m))$ if and only if $\lambda_R(r,\Half(v)) = \lambda_R(r,\Half(v'(i_m,t_m))) = (r_m,\Half(w')\Half(o_m,t_m-t_{m-1})) = (r_m,\Half(w))$, and thus that the claim holds also for $m$.

To conclude the proof of the Lemma it is sufficient to recall that from the definition of behaviour we have that $B_T(v) = w$ if and only if $\lambda_S(s_0,v) = (s_m,w)$ for some state $s_m \in S$. From $(s_0,0) \sim r_0$ (hypothesis of the lemma) we can conclude that $\lambda_R(r_0,\Half(v)) = (r_m,\Half(w))$ and thus that $B_U(\Half(v)) = \Half(w) = \Half(B_T(v))$.
\end{proof}

\begin{theorem}\label{th:abstract-Half-bisim}
A TFSM with timeouts and timed guards $M$ is $\Half$-bisimilar to the abstract FSM $A_M$.
\end{theorem}

\begin{proof}
The relation $\sim = \{((s,x),(s,\langle n,n'\rangle))\mid x \in \langle n,n'\rangle \}$ is a $\Half$-bisimulation for $M$ and $A_M$.
\end{proof}

We can use the above theorem to solve the equivalence problem for TFSM with timed guards.

\begin{corollary}\label{th:tfsm-to-equiv}
Let $M$ and $M'$ be two TFSM with timeouts and timed guards.
Then $M$ and $M'$ are equivalent if and only if the two abstract FSM $A_M$ and $A_{M'}$ are equivalent.
\end{corollary}

\begin{proof}
The claim is a direct consequence of Theorem~\ref{th:abstract-Half-bisim} and Lemma~\ref{lem:Half-bisimilar-then-equiv}. 
\end{proof}

\input{untimed2timed}

\input{intersection-new}

\input{tfsm-ta-biblio}

\section{Conclusions}
\label{sec:conclusions}

We investigated deterministic TFSMs with a single clock, 
with both timed guards and timeouts.
We showed that the behaviours of the timed FSMs are equivalent if and only 
if the behaviours of the companion untimed FSMs obtained by time-abstracting
bisimulations are equivalent, so that they exhibit a good trade-off between
expressive power and ease of analysis.

Then we defined and proved the correctness of the backward construction from
Untimed FSMs to TFSMs. The construction starts from any deterministic FSM 
recognizing a subset of the language 
$\left(\left(\sfrac{\Half}{\Half}\right)^* \sfrac{I}{O}\right)^*\left(\sfrac{\Half}{\Half}\right)^*$
and builds a deterministic TFSM that recognizes the corresponding timed 
language. 
Using the two constructions we showed how to intersect two deterministic TFSMs,
first by transforming them into untimed FSMs, then applying the standard 
intersection algorithm for untimed FSMs, and then transforming back into
a deterministic TFSM.

Future work includes studying more general composition operators
to define and solve equations over deterministic TFSMs~\cite{villa-ucp-book},
and addressing the previous problems for TFSMs with output delays~\cite{Merayo2008} 
and nondeterministic TFSMs.


\subsection*{Acknowledgements}
Davide Bresolin and Tiziano Villa acknowledge partial support from
the project INdAM, GNCS 2020 (Strategic Reasoning and Automated Synthesis of Multi-Agent Systems) funded by MIUR (Italian Ministry of Education, University and Research).
Tiziano Villa was partially supported by MIUR, ``Project Italian Outstanding 
Departments, 2018-2022''.
Nina Yevtushenko was partly supported by the Ministry of Science and Higher Education of the Russian Federation (grant number 075-15-2020-788).


\bibliographystyle{abbrv}     
\bibliography{tfsms}

\end{document}

%% file: macro.tex
\newcommand{\bbI}{\mathbb{I}}
\def\bbR{\mathbb{R}} 
\def\bbN{\mathbb{N}} 

\newcommand{\Half}{\mathbbmss{t}}

\newcommand{\floor}[1]{{\lfloor {#1} \rfloor}}

\DeclareMathOperator{\Untime}{Untime}

\newcommand{\dims}[1]{k}

\newcommand{\trans}[1]{\xrightarrow{#1}}

\newcommand{\cut}[1]{ }

%% file: untimed2timed.tex

\section{From untimed FSMs to TFSMs}\label{sec:untimed2timed}

In the previous section we have shown how to build an abstract untimed FSM that represents the behaviour of a TFSM, by means of appropriate notions of bisimulation and of abstraction of timed words. In this section we study the conditions under which the opposite transformation is possible: we take an untimed FSM that accepts and produces words from input and output alphabets that include the special symbol $\Half$, and we show how to build an equivalent TFSM with timeouts and timed guards, under the same notion of abstraction of timed words.

Now, let $I$ and $O$ be, respectively, the input and output alphabets of our machines. We are interested in studying untimed FSMs that accept words in $(I \cup \{\Half\})^*$ and produce words in $(O \cup \{\Half\})^*$. Clearly, not all untimed FSMs represent valid timed behaviours. In particular, since in our TFSMs model outputs are instantaneously produced when an input is received, and since a TFSM cannot stop the advancing of time, we have that a deterministic untimed FSM $U = (R, I \cup \{\Half\}, O \cup \{\Half\}, \lambda_R, r_0)$ can be transformed into a TFSM only if every state $r$ of $U$ respects the following two conditions:
\begin{compactenum}
	\item $\lambda_R(r, \Half)$ is defined and such that $\lambda_R(r, \Half) = (r', \Half)$ for some $r' \in R$ (when the input $\Half$ is received, the FSM should produce the output $\Half$);
	\item for every input $i \in I$, if $\lambda(r, i)$ is defined then $\lambda(r, i) = (r', o)$ for some output $o \in O$ and state $r' \in R$ (when an input from $I$ is received, the FSM produces an output from $O$).
\end{compactenum}

\noindent We call any untimed FSM that respects the above two conditions \emph{time progressive}.

In the following we prove that every \emph{deterministic time progressive} FSM can be transformed into an equivalent TFSM with timeouts and timed guards.
Since we cannot directly compare the behavior of an untimed FSM with the behavior of a timed FSM, we will use the notion of $\Half$-abstraction of a timed word (Definition~\ref{def:Half-abstract-timed-word}) to compare timed and untimed machines.

\begin{definition}\label{def:Half-refinement-FSM}
Given a deterministic and time progressive FSM $U = (R, I \cup \{\Half\}, O \cup \{\Half\}, \lambda_R,$ $ r_0)$, and a TFSM with timed guards and timeouts  $T = (S, I, O, \lambda_S, s_0, \Delta_S)$, we say that \emph{$T$ refines $U$} if and only if for every timed input word $v = (i_1,t_1)\dots(i_m,t_m)$ we have that $B_U(\Half(v)) = \Half(B_T(v))$.
\end{definition}

The intuition behind the construction is the following. Since we start from a deterministic and time progressive FSM $U$, from every state of $U$ there exists exactly one transition with input $\Half$ (and output $\Half$). Hence, given a state $s$ we can build the (infinite) ``delay run''
$$
\rho_{\Half}^s = s \trans{\Half,\Half} s_1 \trans{\Half,\Half} s_2 \trans{\Half,\Half} \dots
$$

\noindent Since the number of states of $U$ is finite, we have that the delay run is ``lasso shaped'', namely, that it consists of a prefix $s \trans{\Half,\Half} \dots \trans{\Half,\Half} s_p$ followed by the infinite repetition of a loop $s_p \trans{\Half,\Half} \dots \trans{\Half,\Half} s_p$.

The refined TFSM $T$ will have the same set of states of $U$. Then, for every state $s$ the delay run is computed, and the transitions and timeouts are defined as follows:
\begin{compactitem}
	\item every $I/O$ transition leaving a state in the prefix is replaced with a timed transition from $s$ with an appropriate timed guard;
	\item a timeout corresponding to the length of the prefix forces $T$ to switch from $s$ to a state in the loop.
\end{compactitem}

\medskip

\begin{algorithm}[tbp] 
\caption{Transform a FSM into a TFSM with timeouts and timed guards}
\label{alg:main-fsm2tfsm}
\begin{algorithmic}[1]
\Require{A time progressive and deterministic FSM $U = (S, I \cup \{\Half\}, O \cup \{\Half\}, \lambda_U, s_0)$}
\Ensure{A TFSM $T =(S, I, O, \lambda_T, s_0, \Delta_T)$ that refines $U$}
\Function{Refine}{$U$}
	\State{$\lambda_T \gets \emptyset$}
	\State{$\Delta_T \gets \emptyset$}
	\State{$T \gets (S, I, O, \lambda_T, s_0, \Delta_T)$}
	\ForAll{$s \in S$}
		\State{$\Call{AddTimedTrans}{s, U, T}$}
	\EndFor
	\State{\Return{T}}
\EndFunction
\end{algorithmic}
\end{algorithm}

Algorithms \ref{alg:main-fsm2tfsm} and \ref{alg:timed-transitions} describe the above procedure in detail. To simplify the code, we will unfold the final loop once, and put the timeout in correspondence to the second occurrence of $s_p$ in the delay run. Moreover, since $U$ is assumed to be deterministic, we consider the transition relation as a partial function $\lambda_U: S \times I \cup\{\Half\} \mapsto S\times O \cup \{\Half\}$ returning the next state and the output. 

\begin{algorithm}[tbp] 
\caption{Add timed transitions to a state $s$}
\label{alg:timed-transitions}
\begin{algorithmic}[1]
\Function{AddTimedTrans}{$s, U, T$}
	\ForAll{$r \in S$} 
		{$\Call{Marked}{r} \gets False$}
	\EndFor
	\State{$r \gets s$}
	\State{$g \gets [0,0]$}
	\While{\textbf{not} \Call{Marked}{$r$}}\label{ttrans:while}
		\State{$\Call{Marked}{r} \gets True$}
		\ForAll{$i \in I$ \textbf{such that} $i\neq \Half$ and $\lambda_U(r,i)$ is defined}\label{ttrans:forinputs}
			\State{$(r',o) \gets \lambda_U(r,i)$}
			\State{\textbf{add} $(s, i, g, o, r')$ \textbf{to} $\lambda_T$}
		\EndFor\label{ttrans:endforinputs}
		\State{$r \gets \lambda_U(r, \Half)\downarrow_S$}\label{ttrans:nextr}
		\If{$g = [n,n]$}\label{ttrans:ifg}
			{$g \gets (n,n+1)$}
		\ElsIf{$g = (n,n+1)$} {$g \gets [n+1,n+1]$}
			
		\EndIf\label{ttrans:endifg}
	\EndWhile\label{ttrans:endwhile}
	\If{$g = [n,n]$}\label{ttrans:iftimeout}	
		{$\Delta_T(s) = (r,n)$}
		\Comment{prefix of even length: set the timeout and return}
	\ElsIf{$g = (n,n+1)$} \label{ttrans:elsiftimeout}
		\Comment{prefix of odd length: unfold it one more step}
		\ForAll{$i \in I$ \textbf{such that} $i\neq \Half$}
			\State{$(r',o) \gets \lambda_U(r,i)$}
			\State{\textbf{add} $(s, i, g, o, r')$ \textbf{to} $\lambda_T$}
		\EndFor
		\State{$r \gets \lambda_U(r, \Half)\downarrow_S$}
		\State{$\Delta_T(s) = (r,n+1)$}
	\EndIf\label{ttrans:endiftimeout}
\EndFunction
\end{algorithmic}
\end{algorithm}

We prove the correctness of our construction by showing that the TFSM $T$ obtained from Algorithm~\ref{alg:main-fsm2tfsm} is $\Half$-bisimilar to $U$. Hence, by Lemma~\ref{lem:Half-bisimilar-then-equiv}, we can immediately conclude that $T$ is a refinement of $U$.

\begin{theorem}\label{th:timed-refinement-bisim}
Given a time progressive and deterministic FSM $U = (S, I \cup \{\Half\}, O \cup \{\Half\}, \lambda_U,$ $ s_0)$, Algorithm~\ref{alg:main-fsm2tfsm} builds a TFSM with timeouts and timed guards $T =(S, I, O, \lambda_T,$ $s_0, \Delta_T)$ for which there exists a $\Half$-bisimulation $\sim$ such that $(s_0,0) \sim s_0$.
\end{theorem}

\begin{proof}
Let $U = (S, I \cup \{\Half\}, O \cup \{\Half\}, \lambda_U, s_0)$ be a time progressive and deterministic FSM, and let $T =(S, I, O, \lambda_T, s_0, \Delta_T)$ be the TFSM built by Algorithm~\ref{alg:main-fsm2tfsm}. We define the following relation between states of $T$ and states of $U$:
\begin{equation}\label{eq:u2t_bisim}
\sim = \{((s,x), r) \mid r = \hat{\lambda}_U(s, \Half(x))\downarrow_S \}
\end{equation}

\noindent where $\hat{\lambda}_U: S \times (I \cup \{\Half\})^* \mapsto S \times (O \cup \{\Half\})^*$ is the usual extension of the transition function $\lambda_U$ to input words. 

We show that $\sim$ is indeed a $\Half$-bisimulation between $T$ and $U$ by proving that the function \Call{AddTimedTransition}{} (Algorithm~\ref{alg:timed-transitions}) respects the following invariant:
\begin{itemize}
	\item[\bf INV] $(s,x) \sim r$ for all $x \in g$, and all conditions of Definition~\ref{def:Half-bisim} are respected by the transitions in $\lambda_T$
	\end{itemize}

Before entering into the \textbf{while} loop, \Call{AddTimedTransition}{} sets $r = s$ and $g = [0,0]$. Since $\Half(0) = \varepsilon$, we have that $(s,0) \sim s$, and since $\lambda_T$ is empty, Definition~\ref{def:Half-bisim} is trivially respected. 

Consider now a generic iteration of the \textbf{while} loop (lines \ref{ttrans:while}--\ref{ttrans:endwhile}). By the invariant, we have that $(s,x) \sim r$ for all $x \in g$. The \textbf{for} loop (lines \ref{ttrans:forinputs}--\ref{ttrans:endforinputs}) iterates through all transitions of $U$ activated by an actual input $i \in I$, adding a transition $(s,i,o,g,r')$ to $\lambda_T$ for every transition $(r,i,o,r') \in \lambda_U$. Hence, for every $x \in g$ we have that $(s,x) \trans{i,o} (r',0)$ and $r \trans{i,o} r'$. Since $(s,x) \sim r$ and $(r',0) \sim r'$, we have that conditions 3 and 4 of Definition~\ref{def:Half-bisim} are respected. 
After updating $\lambda_T$, lines \ref{ttrans:nextr}--\ref{ttrans:endifg} update the value of $r$ and $g$. Let us call $r_{old}$ and $g_{old}$ the values of $r$ and $g$ before the update. Then, $r$ is set to the $\Half$-successor of $r_{old}$ and $g$ is updated to the ``next interval'' as follows:
\begin{compactitem}
	\item if $g_{old} = [n,n]$ then $g = (n,n+1)$;
	\item if $g_{old} = (n,n+1)$ then $g = [n+1,n+1]$.
\end{compactitem}
We consider the two cases separately. If $g = [n,n]$ then the only possible state $(s,x)$ such that $x \in [n,n]$ is $(s,n)$. Moreover, by the definition of $\sim$, since $(s,n) \sim r_{old}$ we have that $r_{old} = \hat{\lambda}_U(s, \Half(n))\downarrow_S$, with $\Half(n) = \Half^{2n}$. 
Since line \ref{ttrans:nextr} updates $r$ to $\lambda_U(r_{old},\Half)$, and since $\Half(x+t) = \Half^{2n+1}$ for every $0 < t < 1$, we have that $r = \hat{\lambda}_U(s,\Half(x+t))\downarrow_S$. Hence, since $(s,n) \trans{t} (s,n+t)$, $(s,n+t) \sim r$ and $r_{old} \trans{\Half,\Half} r$ we have that conditions 1 and 2 of Definition~\ref{def:Half-bisim} are respected.
By a similar argument, if $g = (n,n+1)$ we can show that $(s,x) \trans{t} (s,x+t)$, $(s,x+t) \sim r$ and $r_{old} \trans{\Half,\Half} r$ for every $x$ and $t$ such that $0 < t < 1$ and $x+t = n+1$, respecting conditions 1 and 2 of Definition~\ref{def:Half-bisim} also in this case. 
Hence, every iteration of the \textbf{while} loop respects the invariant.

The loop terminates when $r$ is a \Call{Marked}{} state, that is, when it reaches the first repetition of a state in the delay run from $s$. Lines \ref{ttrans:iftimeout}--\ref{ttrans:endiftimeout} take care of setting appropriately the timeout at state $s$. Two different situations may arise: either $g = [n,n]$ or $g = (n,n+1)$ for some $n \in \bbN$. In the former case, the state $r$ is repeated after an even number of transitions, which corresponds to an integer time delay. Hence, the timeout at $s$ is set to $\Delta_S(s) = (r, n)$. Consider now the predecessor $r_{pred}$ of $r$ in the delay run. By the invariant, we have that $(s, x) \sim r_{pred}$ for every $x \in (n-1,n)$. Hence, we have that $(s, x) \trans{n - x} (r,0)$ for every $n-1 < x < n$, $r_{pred} \trans{\Half,\Half} r$, $(s, x) \sim r_{pred}$ and $(r,0) \sim r$, respecting conditions 1 and 2 of Definition~\ref{def:Half-bisim}.
In the latter case ($g = (n,n+1)$), $r$ is repeated after an odd number of transitions. Since the timeout at $s$ must be an integer value, lines \ref{ttrans:elsiftimeout}--\ref{ttrans:endiftimeout} repeat the construction of the \textbf{while} loop one more time and then update $r$ to a state that corresponds to precisely $n+1$ time units before setting the timeout. As in the previous case, we can prove that the invariant is respected.

To conclude the proof we observe that Algoritm~\ref{alg:main-fsm2tfsm} executes \Call{AddTimedTransition}{} on every state $s \in S$. Hence, the final TFSM $T$ is in relation $\sim$ with $U$. Since $\sim$ respects all conditions of Definition~\ref{def:Half-bisim}, we have that it is a $\Half$-bisimulation between $T$ and $U$ such that $(s_0,0) \sim s_0$.
\end{proof}

\begin{corollary}\label{cor:timed-refinement}
Given a time progressive and deterministic FSM $U = (S, I \cup \{\Half\}, O \cup \{\Half\}, \lambda_U,$ $ s_0)$, Algorithm~\ref{alg:main-fsm2tfsm} builds a TFSM with timeouts and timed guards \linebreak $T =(S, I, O, \lambda_T, s_0, \Delta_T)$ that refines $U$.
\end{corollary}

\begin{figure}[tbp]
	\centering
	\begin{tikzpicture}[font=\footnotesize,yscale=0.8]
		\node[draw,circle] (q0)	 at (0,0)	{$q_0$};
		\node[draw,circle]	(q2) at (4,0)	{$q_2$};
		\node[draw,circle]	(q5) at (2,-4)	{$q_5$};

		\draw[->] (-0.6,-0.6) -- (q0);
		\path[->] (q0) edge[loop above] node        {$[0,1):i/o_1$} (q0);
		\path[->] (q0) edge[loop left] node         {$(2,3):i/o_1$} (q0);
		\path[->] (q0) edge[bend left=10]  node[above] {$[1,2]:i/o_2$} (q2);
		\path[->] (q0) edge[bend left=10]  node[above right] {$t=3$} (q5);

		\path[->] (q2) edge[loop above] node        {$[0,1]:i/o_2$} (q1);
		\path[->] (q2) edge[bend left=10]  node[below] {$(1,2):i/o_1$} (q0);
		\path[->] (q2) edge						  node[below right] {$t=2$} (q5);

		\path[->] (q5) edge[loop below] node        {$t=1$} (q5);
		\path[->] (q5) edge[bend left=10]  node[below left] {$[0,1):i/o_1$} (q0);
	\end{tikzpicture}
	\caption{Example of application of Algorithm~\ref{alg:main-fsm2tfsm}.}
  \label{fig:fsm2tfsm}
\end{figure}
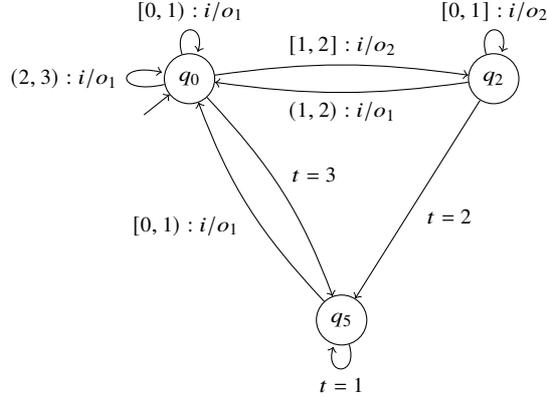

Figure~\ref{fig:fsm2tfsm} shows the TFSM with timeouts and timed guards that can be obtained by applying Algorithm~\ref{alg:main-fsm2tfsm} to the untimed FSM of Figure~\ref{fig:tfsm-all-ab}, where the states have been renamed as follows:
\begin{align*}
	(s_0,[0,0]) & = q_0 	&		(s_0,(0,1)) & = q_1			&			(s_1,[0,0]) & = q_2		\\
	(s_1,(0,1)) & = q_3 	&		(s_1,[1,1]) & = q_4			&			(s_1,(1,\infty)) & = q_5		 
\end{align*}

\noindent \sloppy In the picture, transitions with adjacent guards have been merged: for instance, the application of the algorithm creates the transitions $(q_0, i, o_1, [0,0], q_0)$ and the transition $(q_0, i, o_1, (0,1), q_0)$ that are merged into the unique transition $(q_0, i, o_1, [0,1), q_0)$ in the picture. 
The picture includes only the states that are reachable from the initial state $q_0$.
This shows that in the final result only the three states $q_0$, $q_2$ and $q_5$ are relevant: the other states have been replaced by either timed guards or timeouts. 

To better understand how Algorihm~\ref{alg:main-fsm2tfsm} works, let us review the application of function \Call{AddTimedTrans}{} (Algorithm~\ref{alg:timed-transitions}) to the initial state $q_0$ (state $(s_0,[0,0])$ in the picture) of the untimed FSM $A_M$ of Figure~\ref{fig:tfsm-all-ab}. The procedure starts by unmarking all states of $A_M$ and by initialising the current state $r$ to $q_0$ and the current guard $g$ to $[0,0]$. Then the \textbf{while} loop of lines \ref{ttrans:while}--\ref{ttrans:endwhile} follows the sequence of $\Half/\Half$ transitions in $A_M$, marking the states it reaches, until a previously marked state is found. At lines 7-9, for every I/O transition exiting the current state, a corresponding timed transition labelled with the current value of $g$ is added to the TFSM. Then the current state $r$ is updated to the next state in the sequence of $\Half/\Half$ transitions and $g$ is increased following the sequence $[0,0], (0,1), [1,1], (1,2), \dots$. 
In this example, the first iteration of the \textbf{while} loop considers all I/O transitions exiting from the state $q_0$ of $A_M$, namely the transition $q_0 \trans{i/o_1} q_0$, and adds the transition $q_0 \trans{[0,0]:i/o_1} q_0$ to the TFSM (the initial value of $g$ is indeed $[0,0]$). Then $r$ is updated to $q_1$, $g$ to $(0,1)$ and the second iteration is started. The transition $q_1 \trans{i/o_1} q_0$ corresponds to the transition $q_0 \trans{(0,1):i/o_1} q_0$ in the TFSM. Notice that the starting state of the timed transition is still $q_0$. The $\Half/\Half$ transition between $q_0$ and $q_1$ of $A_M$ models the fact that the machine waits for a time included in the interval $(0,1)$ before accepting an input. This situation is modelled in the TFSM by adding the guard $(0,1)$ to the transition while keeping $q_0$ as starting state. Then the loop continues by adding the following transitions to the TFSM:
\begin{align*}
 q_0 \trans{[1,1]:i/o_2} q_2 & & q_0 \trans{(1,2):i/o_2} q_2 & & 
 q_0 \trans{[2,2]:i/o_2} q_2 & & q_0 \trans{(2,3):i/o_1} q_0 
\end{align*}

\noindent At this point, the current state $r$ of $A_M$ is $q_5$ (i.e., $(s_1,(1,\infty))$) and the guard $g$ is $(2,3)$. Because of the self loop on $\Half/\Half$ of $A_M$ in state $q_5$, at the end of the loop $r$ does not change and $g$ is updated to $[3,3]$: a previously marked state is reached and the loop terminates. Lines \ref{ttrans:iftimeout}--\ref{ttrans:endiftimeout} of \Call{AddTimedTrans}{} set the timeout at state $q_0$ to $(q_5, 3)$, terminating the function call. The value of the timeout is set to $3$ because the first marked state is reached after $6$ $\Half/\Half$ transitions, which corresponds to $3$ time units.
A subsequent call to \Call{AddTimedTrans}{} on state $q_5$ will set the timeout at state $q_5$ to $(q_5, 1)$ (i.e., the self-loop on $t=1$ depicted in the figure), to model the fact that in the untimed FSM $A_M$ there is a self-loop on $\Half/\Half$ at state $q_5$. In this way, the sequence of $\Half/\Half$ transitions $q_0 \trans{\Half/\Half} q_1 \trans{\Half/\Half} q_2 \trans{\Half/\Half} q_3 \trans{\Half/\Half} q_4 \trans{\Half/\Half} q_5 \trans{\Half/\Half} q_5 \trans{\Half/\Half} q_5 \trans{\Half/\Half} q_5 \trans{\Half/\Half} q_5 \trans{\Half/\Half} \dots$ of $A_M$ is replaced by the sequence of timeout transitions $q_0 \trans{3} q_5 \trans{1} q_5 \trans{1} \dots$. In both cases the machines can wait in $q_5$ forever, if no input is received in the first $3$ time units.
The application of \Call{AddTimedTrans}{} to the other states of $A_M$ builds the rest of the TFSM.

\begin{figure}
\centering
\begin{tikzpicture}[scale=1.2,font=\footnotesize,auto]
		
		\node[draw,circle,minimum size=20pt,text width=4ex,text centered] (q000)	 at (0,0)	{$q_0$ \\ $[0,0]$};
		\node[draw,circle,minimum size=20pt,text width=4ex,text centered] (q001)	 at (4,0)	{$q_0$ \\ $(0,1)$};
		\node[draw,circle,minimum size=20pt,text width=4ex,text centered] (q011)	 at (8,0)	{$q_0$ \\ $[1,1]$};
		\node[draw,circle,minimum size=20pt,text width=4ex,text centered] (q012)	 at (8,-3.5)	{$q_0$ \\ $(1,2)$};
		\node[draw,circle,minimum size=20pt,text width=4ex,text centered] (q022)	 at (5,-3.5)	{$q_0$ \\ $[2,2]$};
		\node[draw,circle,minimum size=20pt,text width=4ex,text centered] (q023)	 at (2.5,-3.5)	{$q_0$ \\ $(2,3)$};

		\node[draw,circle,minimum size=20pt,text width=4ex,text centered] (q200)	 at (7,-1.5)	{$q_2$ \\ $[0,0]$};
		\node[draw,circle,minimum size=20pt,text width=4ex,text centered] (q201)	 at (5.5,-1.5)	{$q_2$ \\ $(0,1)$};
		\node[draw,circle,minimum size=20pt,text width=4ex,text centered] (q211)	 at (4,-1.5)	{$q_2$ \\ $[1,1]$};
		\node[draw,circle,minimum size=20pt,text width=4ex,text centered] (q212)	 at (2.5,-1.5)	{$q_2$ \\ $(1,2)$};

		\node[draw,circle,minimum size=20pt,text width=4ex,text centered] (q500)	 at (1.5,-2.5)	{$q_5$ \\ $[0,0]$};
		\node[draw,circle,minimum size=20pt,text width=4ex,text centered] (q501)	 at (-0,-2.5)	{$q_5$ \\ $(0,1)$};

		\path[->] (-0.75,0) edge (q000);
		
		\path[->] (q000) edge node[below] {$\Half/\Half$} (q001);
		\path[->] (q001) edge node {$\Half/\Half$} (q011);
		\path[->] (q011) edge node {$\Half/\Half$} (q012);
		\path[->] (q012) edge node {$\Half/\Half$} (q022);
		\path[->] (q022) edge node {$\Half/\Half$} (q023);
		\path[->] (q023) edge node {$\Half/\Half$} (q500);

		\path[->] (q200) edge node {$\Half/\Half$} (q201);
		\path[->] (q201) edge node {$\Half/\Half$} (q211);
		\path[->] (q211) edge node {$\Half/\Half$} (q212);
		\path[->] (q212) edge node {$\Half/\Half$} (q500);

		\path[->] (q500) edge[bend left=10] node {$\Half/\Half$} (q501);
		\path[->] (q501) edge[bend left=10] node {$\Half/\Half$} (q500);

		\path[->] (q000) edge[loop above] node {$i/o_1$} (q000);
		\path[->] (q001) edge[bend right=15] node[above] {$i/o_1$} (q000);
		\path[->] (q011) edge node[below right] {$i/o_2$} (q200);
		\path[->] (q012) edge node {$i/o_2$} (q200);
		\path[->] (q022) edge node[below right] {$i/o_2$} (q200);
		\path[->] (q023) edge[out=180,in=-120,looseness=1.75] node[near start] {$i/o_1$} (q000);
		\path[->] (q200) edge[loop above] node {$i/o_2$} (q200);
		\path[->] (q201) edge[bend left=45] node[above,near start] {$i/o_2$} (q200);
		\path[->] (q211) edge[bend left=55] node[above] {$i/o_2$} (q200);
		\path[->] (q212) edge node[near start,above right] {$i/o_1$} (q000);
		\path[->] (q500) edge node {$i/o_1$} (q000);
		\path[->] (q501) edge node {$i/o_1$} (q000);
		
		\pgfresetboundingbox
		\path (-1.25,-4) rectangle (8.5,1.25);
	\end{tikzpicture}
	\caption{$\Half$-abstraction of the TFSM in Figure~\ref{fig:fsm2tfsm}.}
  \label{fig:fsm2tfsm2fsm}
 \end{figure}
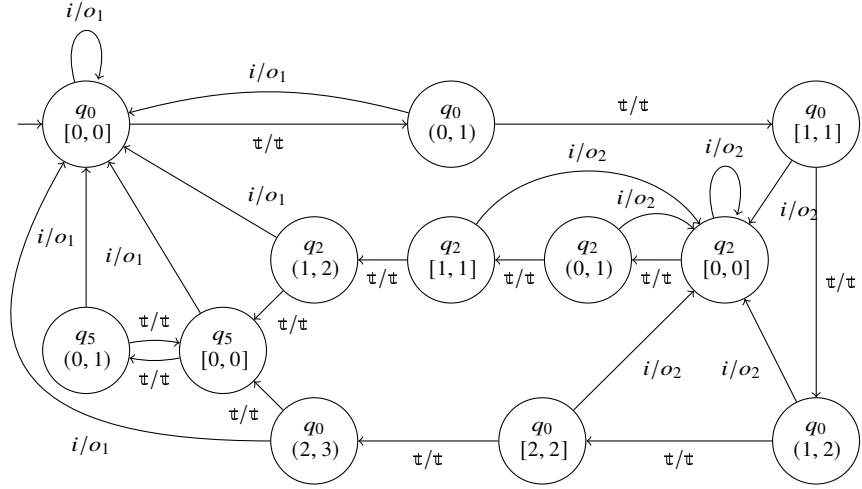

By applying the equivalence checking methodology presented in Section~\ref{sec:timedfsms}, we can prove that the TFSM of Figure~\ref{fig:fsm2tfsm} is indeed equivalent to the TFSM of Figure~\ref{fig:tfsm-all-ex}. 
Figure~\ref{fig:fsm2tfsm2fsm} shows the $\Half$-abstraction of the TFSM of Figure~\ref{fig:fsm2tfsm}, which is equivalent to the FSM of Figure~\ref{fig:tfsm-all-ab} (by standard FSM state-minimization of the FSM in Figure~\ref{fig:fsm2tfsm2fsm}, we get a reduced FSM isomorphic to the one in Figure~\ref{fig:tfsm-all-ab}). 
This is consistent with the fact that the FSM of Figure~\ref{fig:tfsm-all-ab} is the $\Half$-abstraction of the TFSM of Figure~\ref{fig:tfsm-all-ex}.

%% file: intersection-new.tex

\section{Intersection of TFSMs}
\label{sec:intersection-new}

In this section we apply the previous transformations to perform
the intersection of TFSMs.
In general, TFSMs can be composed to build complex systems out of simpler 
components.
Several composition operators exist for untimed FSMs, the most relevant ones 
being the intersection operator, the serial composition, and synchronous and asynchronous parallel composition (see~\cite{langeq-procieee2015}). 
Parallel composition of TA was discussed in~\cite{sifakis-stacs96}.
Preliminary work on parallel composition of TFSMs with timed guards and output
delays can be found in~\cite{DBLP:conf/qsic/KondratyevaKCY13}, and
on parallel composition of TFSMs with timeouts and output delays
in~\cite{gromov-ewdts2016}. 
When extending compositions to Timed FSMs, one must verify that TFSMs are closed under the type of composition of interest. 
In our setting, this means that the behaviour of the composed system should be represented by a machine with only a single clock. 
Here we focus on the intersection operator for which we show that closure holds.

In the following we show how the transformation from TFSMs to untimed FSMs of Section~\ref{sec:timedfsms} and the transformation from untimed FSMs to TFSMs of Section~\ref{sec:untimed2timed} can be used to implement the intersection of TFSMs. 
Suppose that we have two TFSMs $M_1$ and $M_2$ and that we want to compute the intersection $M_1 \cap M_2$ whose behaviour is the intersection of the behaviours of $M_1$ and $M_2$. We can proceed as follows:

\begin{compactenum}
	\item compute the $\Half$-abstract FSMs $A_{M_1}$ and $A_{M_2}$ as in Definition~\ref{def:abstract-tfsm-all} for, respectively, $M_1$ and $M_2$;
	\item intersect $A_{M_1}$ and $A_{M_2}$ using the standard algorithm for untimed FSMs, obtaining the untimed FSM $C = A_{M_1} \cap A_{M_2}$;
	\item compute the TFSM $T$ that is $\Half$-bisimilar with $C$ using Algorithm~\ref{alg:main-fsm2tfsm}.
	\end{compactenum}

\noindent The following theorem shows that $T$ is equivalent to the intersection of $M_1$ and $M_2$. 

\begin{theorem}\label{th:intersection}
Let $M_1$ and $M_2$ be two deterministic TFSMs, and let $T = \Call{Refine}{A_{M_1} \cap A_{M_2}}$. Then, for every timed input word $v = (i_1,t_1)\dots (i_k, t_k)$ we have that $B_T(v) = w = (o_1,t_1)\dots (o_k,t_k)$ if and only if $B_{M_1}(v)$ and $B_{M_2}(v)$ are defined and such that $B_{M_1}(v) = B_{M_2}(v) = w$.
\end{theorem}

\begin{proof}
Let $M_1$ and $M_2$ be two deterministic TFSMs, and let $A_{M_1}$ and $A_{M_2}$ be their respective $\Half$-abstractions. By Definition~\ref{def:abstract-tfsm-all} we have that $A_{M_1}$ and $A_{M_2}$ are deterministic and time progressive. Hence, the intersection $A_{M_1} \cap A_{M_2}$ is also deterministic and time progressive and  Algorithm~\ref{alg:main-fsm2tfsm} can be applied to obtain the TFSM $T$.

To prove the direct implication, let $v = (i_1,t_1)\dots (i_k, t_k)$ be an input timed word and suppose that $B_T(v) = w$ for some timed output word $w = (o_1,t_1)\dots (o_k,t_k)$. Since $T = \Call{Refine}{A_{M_1} \cap A_{M_2}}$, by Corollary~\ref{cor:timed-refinement} we have that $T$ refines $A_{M_1} \cap A_{M_2}$. Hence, by Definition~\ref{def:Half-refinement-FSM} we have that $B_{A_{M_1} \cap A_{M_2}}(\Half(v)) = \Half(B_T(v)) = \Half(w)$. Since $A_{M_1} \cap A_{M_2}$ is the intersection of $A_{M_1}$ and $A_{M_2}$, we have that $B_{A_{M_1}}(\Half(v)) =  B_{A_{M_2}}(\Half(v)) = \Half(w)$.
Since $A_{M_1}$ and $A_{M_2}$ are the $\Half$-abstraction of $M_1$ and $M_2$, by Theorem~\ref{th:abstract-Half-bisim} and Lemma~\ref{lem:Half-bisimilar-then-equiv} we have that $\Half(w) = B_{A_{M_1}}(\Half(v))) = \Half(B_{M_1}(v))$ and $\Half(w) = B_{A_{M_2}}(\Half(v))) = \Half(B_{M_2}(v))$. This proves that $B_{M_1}(v)$ and $B_{M_2}(v)$ are defined and such that $B_{M_1}(v) = B_{M_2}(v) = w$.

To prove the opposite implication, let $v = (i_1,t_1)\dots (i_k, t_k)$ be an input timed word and suppose that $B_{M_1}(v)$ and $B_{M_2}(v)$ are defined and such that $B_{M_1}(v) = B_{M_2}(v) = w$ for some timed output word $w = (o_1,t_1)\dots (o_k,t_k)$. Since $A_{M_1}$ and $A_{M_2}$ are the $\Half$-abstraction of $M_1$ and $M_2$, by Theorem~\ref{th:abstract-Half-bisim} and Lemma~\ref{lem:Half-bisimilar-then-equiv} we have that $B_{A_{M_1}}(\Half(v))) = \Half(B_{M_1}(v)) = \Half(w)$ and $B_{A_{M_2}}(\Half(v))) = \Half(B_{M_2}(v)) = \Half(w)$.
Hence, the intersection $A_{M_1} \cap A_{M_2}$ is such that $B_{A_{M_1} \cap A_{M_2}}(\Half(v)) = B_{A_{M_1}}(\Half(v)) =  B_{A_{M_2}}(\Half(v)) = \Half(w)$.
Since $T = \Call{Refine}{A_{M_1} \cap A_{M_2}}$, by Corollary~\ref{cor:timed-refinement} and Definition~\ref{def:Half-refinement-FSM} we have that $\Half(B_T(v)) = B_{A_{M_1} \cap A_{M_2}}(\Half(v)) =  \Half(w)$. Hence, we have proved that $B_T(v) = w$.
\end{proof}

As an example, consider the TFSMs $M_1$ and $M_2$ of Figure~\ref{fig:intersection_start}, and suppose we want to compute the \emph{intersection} $M_1 \cap M_2$.
Following the above procedure, the first step is to obtain the $\Half$-abstract FSMs $A_{M_1}$ and $A_{M_2}$ in Figure~\ref{fig:intersection_step1}. Then, by applying the standard constructions for intersection and minimization of untimed FSMs, we obtain the machine $C$ depicted in Figure~\ref{fig:intersection_step2} and finally, using Algorithm~\ref{alg:main-fsm2tfsm}, the TFSM $T = \Call{Refine}{A_{M_1} \cap A_{M_2}}$ of Figure~\ref{fig:intersection_final}. 
It is worth pointing out that the intersection of two complete and deterministic TFSMs is still a deterministic machine, but it may be partial. This is indeed the case of our example: for instance, when the TFSM in Figure~\ref{fig:intersection_final} is in state $0$ it can react to the input $i$ only when the clock is in the intervals $[0,0]$ or $(1,2)$. No behaviour is specified when the clock is inside the interval $(0,1]$ and $[2,3)$.
In state $1$ and $13$ no behaviour is specified when the clock has an integer value smaller than the timeout ($0, 1, 2$ and $3$ for state $1$, $0$ for state $13$).

\begin{figure}[tbp]
\ \hfill
	\subfigure{
	\begin{tikzpicture}[auto]
		\node	(M1) at (-1.25,0) {$M_1$};
		\node[draw,circle] (A)	 at (0,0)	{$A$};
		\node[draw,circle]	(B) at (3,0)	{$B$};
		\node[draw,circle]	(C) at (1.5,-2.5)	{$C$};
		
		\draw[->] (-0.75,0) -- (A);
		\path[->] (A) edge[loop above] node[above]        {$[1,2):i/o_2$} ();
		\path[->] (A) edge[bend left=15]  node[above right] {$t=2$} (C);
		\path[->] (A) edge					  node[above] {$[0,1):i/o_1$} (B);
		\path[->] (B) edge[loop above] node        {$[0,\infty):i/o_2$} ();
		\path[->] (C) edge[loop right] node        {$(0,1):i/o_1$} ();
		\path[->] (C) edge[bend left=15] node[below left,align=center]        {$[0,0]:i/o_2$ \\ t=1} (A);
	\end{tikzpicture}}
\hfill
	\subfigure{
	\begin{tikzpicture}
		\node	(M2) at (-1.5,0) {$M_2$};
		\node[draw,circle] (a)	 at (0,0)	{$a$};

		\draw[->] (-0.75,0) -- (q0);
		\path[->] (a) edge[loop above] node        {$[0,0]: i/o_1$} ();
		\path[->] (a) edge[loop right] node        {$(0,1): i/o_2$} ();
		\path[->] (a) edge[loop below] node        {$t=1$} ();
		\path (-2,-2.25)--(3,1.25);
	\end{tikzpicture}}
\hfill\ 
	\caption{TFSMs $M_1$ and $M_2$ to be intersected.}
  \label{fig:intersection_start}
\end{figure}
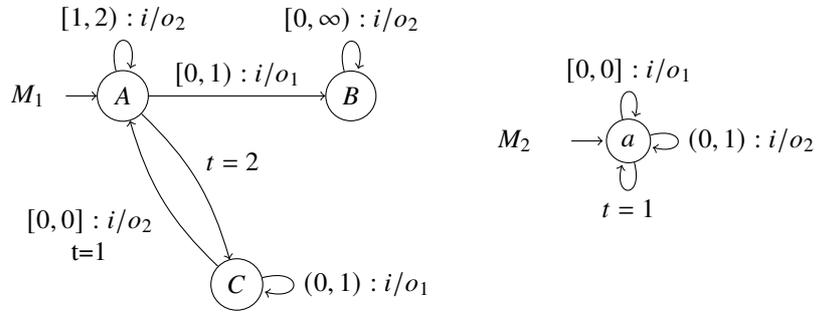

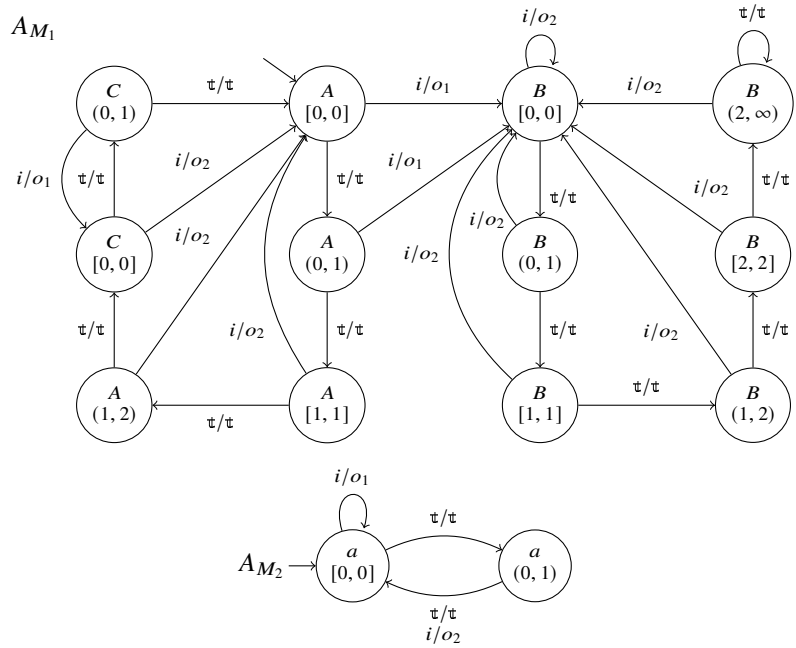
\begin{figure}[tbp]
\centering
	\subfigure{
	\begin{tikzpicture}[auto,xscale=1.4,yscale=1,font=\scriptsize]
		\node	(M1) at (-2.75,1) {\normalsize$A_{M_1}$};
		\node[draw,circle,align=center,inner sep=2pt] (A00)	 at (0,0)	{$A$\\$[0,0]$};
		\node[draw,circle,align=center,inner sep=2pt] (A01)	 at (0,-2)	{$A$\\$(0,1)$};
		\node[draw,circle,align=center,inner sep=2pt] (A11)	 at (0,-4)	{$A$\\$[1,1]$};
		\node[draw,circle,align=center,inner sep=2pt] (A12)	 at (-2,-4)	{$A$\\$(1,2)$};

		\node[draw,circle,align=center,inner sep=2pt] (C00)	 at (-2,-2)	{$C$\\$[0,0]$};
		\node[draw,circle,align=center,inner sep=2pt] (C01)	 at (-2,0)	{$C$\\$(0,1)$};

		\node[draw,circle,align=center,inner sep=2pt] (B00)	 at (2,0)		{$B$\\$[0,0]$};
		\node[draw,circle,align=center,inner sep=2pt] (B01)	 at (2,-2)	{$B$\\$(0,1)$};
		\node[draw,circle,align=center,inner sep=2pt] (B11)	 at (2,-4)	{$B$\\$[1,1]$};
		\node[draw,circle,align=center,inner sep=2pt] (B12)	 at (4,-4)	{$B$\\$(1,2)$};
		\node[draw,circle,align=center,inner sep=2pt] (B22)	 at (4,-2)	{$B$\\$[2,2]$};
		\node[draw,circle,align=center,inner sep=2pt] (B2inf)	 at (4,0)	{$B$\\$(2,\infty)$};

		\draw[->] (-0.6,0.6) -- (A00);
		\path[->] (A00) edge node       {$\Half/\Half$} (A01);
		\path[->] (A01) edge node       {$\Half/\Half$} (A11);
		\path[->] (A11) edge node       {$\Half/\Half$} (A12);
		\path[->] (A12) edge node       {$\Half/\Half$} (C00);

		\path[->] (C00) edge node       {$\Half/\Half$} (C01);
		\path[->] (C01) edge node       {$\Half/\Half$} (A00);

		\path[->] (B00) edge node[near end]       {$\Half/\Half$} (B01);
		\path[->] (B01) edge node       {$\Half/\Half$} (B11);
		\path[->] (B11) edge node       {$\Half/\Half$} (B12);
		\path[->] (B12) edge node[right]       {$\Half/\Half$} (B22);
		\path[->] (B22) edge node[right]       {$\Half/\Half$} (B2inf);
		\path[->] (B2inf) edge[loop above] node       {$\Half/\Half$} (B2inf);

		\path[->] (A00) edge node       {$i/o_1$} (B00);
		\path[->] (A01) edge node       {$i/o_1$} (B00);
		\path[->] (A11) edge[bend left=25] node[near start] {$i/o_2$} (A00);
		\path[->] (A12) edge node {$i/o_2$} (A00);

		\path[->] (C00) edge node       {$i/o_2$} (A00);
		\path[->] (C01) edge[bend right=37.5] node[left]       {$i/o_1$} (C00);

		\path[->] (B00) edge[loop above] node       {$i/o_2$} (B00);
		\path[->] (B01) edge[bend left] node[near start]       {$i/o_2\!\!\!\!$} (B00);
		\path[->] (B11) edge[bend left=37.5] node       {$i/o_2$} (B00);
		\path[->] (B12) edge node[near start]       {$i/o_2$} (B00);
		\path[->] (B22) edge node[near start, above right]       {$i/o_2$} (B00);
		\path[->] (B2inf) edge node[above]       {$i/o_2$} (B00);

	\end{tikzpicture}}
%
	\subfigure{
	\begin{tikzpicture}[auto,xscale=1.2,yscale=1,font=\scriptsize]
		\node	(M2) at (-1,0) {\normalsize$A_{M_2}$};
		\node[draw,circle,align=center,inner sep=2pt] (a00)	 at (0,0)	{$a$\\$[0,0]$};
		\node[draw,circle,align=center,inner sep=2pt] (a01)	 at (2,0)	{$a$\\$(0,1)$};

		\draw[->] (-0.7,0) -- (a00);
		\path[->] (a00) edge[loop above] node        {$i/o_1$} ();
		\path[->] (a00) edge[bend left] node        {$\Half/\Half$} (a01);
		\path[->] (a01) edge[bend left] node[align=center]        {$\Half/\Half$ \\ $i/o_2$} (a00);
	\end{tikzpicture}}

	\caption{Untimed abstractions of $M_1$ and $M_2$.}
  \label{fig:intersection_step1}
\end{figure}

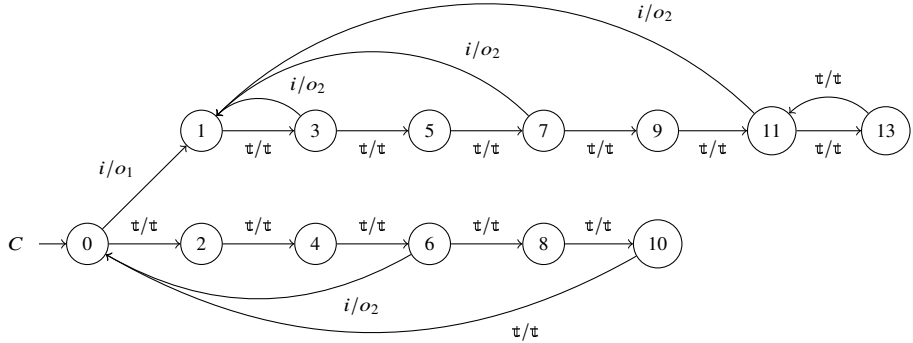
\begin{figure}[tbp]
\centering
	\begin{tikzpicture}[auto,scale=0.75,font=\scriptsize]
		\node	(M1) at (-1.25,0) {$C$};
		\node[draw,circle] (0)	 at (0,0)	{$0$}; 
		\node[draw,circle]	(1) at (2,2)	{$1$}; 
		\node[draw,circle]	(3) at (4,2)	{$3$}; 
		\node[draw,circle]	(5) at (6,2)	{$5$}; 
		\node[draw,circle]	(7) at (8,2)	{$7$}; 
		\node[draw,circle]	(9) at (10,2)	{$9$}; 
		\node[draw,circle]	(11) at (12,2)	{$11$}; 
		\node[draw,circle]	(13) at (14,2)	{$13$}; 

		\node[draw,circle]	(2) at (2,0)	{$2$}; 
		\node[draw,circle]	(4) at (4,0)	{$4$}; 
		\node[draw,circle]	(6) at (6,0)	{$6$}; 
		\node[draw,circle]	(8) at (8,0)	{$8$}; 
		\node[draw,circle]	(10) at (10,0)	{$10$}; 
		
		\draw[->] (-0.85,0) -- (0);
		\path[->] (1) edge  node[below] {$\Half/\Half$} (3);
		\path[->] (3) edge  node[below] {$\Half/\Half$} (5);
		\path[->] (5) edge  node[below] {$\Half/\Half$} (7);
		\path[->] (7) edge  node[below] {$\Half/\Half$} (9);
		\path[->] (9) edge  node[below] {$\Half/\Half$} (11);
		\path[->] (11) edge  node[below] {$\Half/\Half$} (13);
		\path[->] (13) edge[bend right=45]  node[above] {$\Half/\Half$} (11);

		\path[->] (0) edge  node {$\Half/\Half$} (2);
		\path[->] (2) edge  node {$\Half/\Half$} (4);
		\path[->] (4) edge  node {$\Half/\Half$} (6);
		\path[->] (6) edge  node {$\Half/\Half$} (8);
		\path[->] (8) edge  node {$\Half/\Half$} (10);
		\path[->] (10) edge[bend left]  node[near start] {$\Half/\Half$} (0);

		\path[->] (0) edge  node {$i/o_1$} (1);
		\path[->] (3) edge[bend right=45]  node[near start,above right] {$i/o_2$} (1);
		\path[->] (7) edge[bend right=45]  node[near start,above right] {$i/o_2$} (1);
		\path[->] (11) edge[bend right=45]  node[near start,above right] {$i/o_2$} (1);
		\path[->] (6) edge[bend left]  node[near start] {$i/o_2$} (0);
	\end{tikzpicture}
	\caption{The intersection of $A_{M_1}$ and $A_{M_2}$.}
  \label{fig:intersection_step2}
\end{figure}

\begin{figure}[tbp]
\centering
	\begin{tikzpicture}[auto,font=\scriptsize]
		\node	(M1) at (-1.75,0) {$M_1 \cap M_2$};
		\node[draw,circle] (0)	 at (0,0)	{$0$};
		\node[draw,circle]	(1) at (2.5,0)	{$1$};
		\node[draw,circle]	(13) at (5,0)	{$13$};
		
		\draw[->] (-0.75,0) -- (0);

		\path[->] (0) edge  node {$[0,0]: i/o_1$} (1);
		\path[->] (0) edge[loop above]  node {$(1,2):i/o_2$} (0);
		\path[->] (0) edge[loop below]  node {$t=3$} (0);

		\path[->] (1) edge[loop above]  node[align=center] 
				{$(0,1):i/o_2$ \\ $(1,2):i/o_2$ \\ $(2,3):i/o_2$ \\ $(3,4):i/o_2$} (1);
		\path[->] (1) edge node {$t=4$} (13);
		\path[->] (13) edge[bend left] node {$(0,1):i/o_2$} (1);
		\path[->] (13) edge[loop above] node {$t=1$} (13);

	\end{tikzpicture}
	\caption{The TFSM for ${M_1}\cap{M_2}$.}
  \label{fig:intersection_final}
\end{figure}
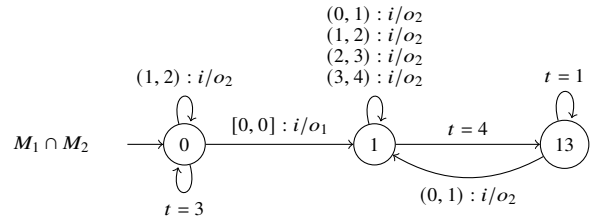

\medskip

%% file: tfsm-ta-biblio.tex
\section{Timed FSMs and Timed Automata}
\label{sec:tfsm-ta}

In this section, we compare TFSMs with Timed Automata (TA), and survey
the known results on the expressivity and computability of various classes
of TA, according to their computational resources.
The landscape of finite automata augmented with time is much more 
complex than in the case of untimed ones, where both language recognizers
(FA) and producers (FSMs) share the fact that there is an underlying common
model with corresponds to regular languages (FSMs transform regular input 
languages into regular output languages).
TA are the most common formalism obtained by adding timing constraints 
(as clocks) to finite-state automata ~\cite{Alur-tcs1994}, defining timed 
regular recognizers.
TA are a more expressive model than TFSMs because they allow
multiple clocks, invariants as conditions on clocks associated to a location,
guards as conditions on clocks associated to a transition,
resets by which a clock may be reset to $0$ or may be kept unchanged,
and states which are products of a location and clock valuations.
Excellent surveys about the classes of TA proposed in the literature can
be found in~\cite{Fontana:2014:MTA:2578702.2518102,ta-csr2013-rudie}.

TFSMs can be transformed into TA with $\varepsilon$-transitions
(called also in the literature silent transitions or internal transitions 
or non-observable transitions) by the following transformation:
\begin{itemize} 
\item there is one location of the TA for every state of the TFSM;
\item given the input and output alphabets $I$ and $O$ of the TFSM,
the alphabet of the TA is given by $I \times O$
\item as in the TFSM, the TA has a single clock, reset to zero 
at every transition;
\item intervals on transitions are replaced with guards;
\item timeouts of the TFSM are replaced by invariants and $\varepsilon$-
transitions.
\end{itemize}
An example of such transformation is shown in Fig.~\ref{fig:TFSM-to-eTA},
where on the left there is a TFSM and on the right the corresponding TA.

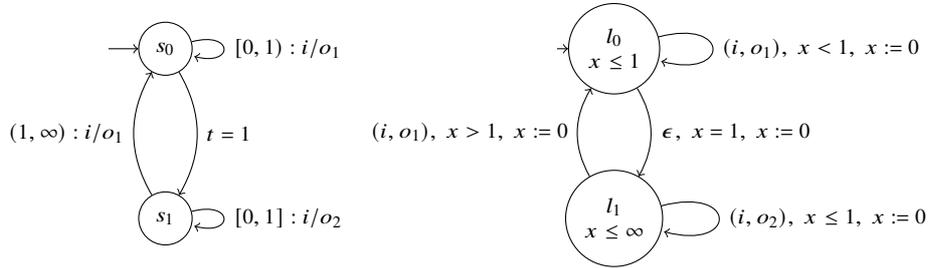
\begin{figure}[tbp]
                \subfigure{
		\begin{tikzpicture}[font=\footnotesize,scale=0.75]
			\path (-3,-4)--(3,1);
			\node[draw,circle,minimum size=20pt] (q0)	 at (0,0)	{$s_0$};
			\node[draw,circle,minimum size=20pt]	(q1) at (0,-3)	{$s_1$};
					
			\draw[->] (-1,0) -- (q0);
			\path[->] (q0) edge[loop right] node        {$[0,1):i/o_1$} ();
			\path[->] (q0) edge[bend left]  node[right] {$t = 1$} (q1);
			\path[->] (q1) edge[bend left]  node[left] {$(1,\infty):i/o_1$} (q0);
			\path[->] (q1) edge[loop right] node        {$[0,1]:i/o_2$} ();
		\end{tikzpicture}}
\hfill
                \subfigure{
		\begin{tikzpicture}[font=\footnotesize,scale=0.75]
			\path (-3,-4)--(3,1);
			\node[draw,circle,minimum size=20pt,align=center] (q0)	 at (0,0)	{$l_0$\\ $x \leq 1$};
			\node[draw,circle,minimum size=20pt,align=center]	(q1) at (0,-3)	{$l_1$ \\ $x \leq \infty$};
					
			\draw[->] (-1,0) -- (q0);
			\path[->] (q0) edge[loop right] node        {$(i,o_1),\ x < 1,\ x:= 0$} ();
			\path[->] (q0) edge[bend left]  node[right] {$\epsilon,\ x = 1,\ x := 0$} (q1);
			\path[->] (q1) edge[bend left]  node[left] {$(i,o_1),\ x > 1,\ x:= 0$} (q0);
			\path[->] (q1) edge[loop right] node        {$(i,o_2),\ x \leq 1,\ x:= 0$} ();
		\end{tikzpicture}}
\caption{Transformation from TFSM (on the left) to $\epsilon$-timed automaton (on the right).}
\label{fig:TFSM-to-eTA}
\end{figure}

This reduction is not necessarily practical, since decision problems are 
in general undecidable for timed automata, even for restricted versions 
of them. In the following we mention some of these relevant results. 
For a classic survey on decision problems for timed automata, see~\cite{alur-decTA2004},
 where the following results can be found:
\begin{enumerate}
\item
TA are closed under union, intersection, projection, but not under
complementation.
\item
The language emptiness problems is PSPACE-complete (a by-product
of reachability analysis obtained by means of the region construction).
\item
The universality. inclusion and equivalence problems for TA
are undecidable.
\item 
Deterministic TA are closed under union. intersection and complementation,
but not under projection. The language emptiness, universality, inclusion
and equivalence problems for deterministic TA are PSPACE-complete.
\end{enumerate}
Further results are proved in~\cite{DBLP:journals/eatcs/Finkel05} and~\cite{DBLP:conf/formats/Finkel06},
e.g., that one cannot decide whether a given timed automaton is 
determinizable or whether the complement of a timed regular
language is timed regular. 

One may wonder whether the complexity goes down, if we reduce the resources 
of the timed automaton. The answer is sometimes yes, but only in very 
restricted cases.
%
In~\cite{Ouaknine:2004:LIP:1018438.1021842,ouaknine-alp2005} it is shown that
the problem of checking language inclusion $L(A) \subseteq L(B)$ 
of TA $A$ and $B$ is decidable if $B$ has no $\epsilon$-transitions.
and either $B$ has only one clock, or the guards of $B$ use only the constant 0.
These two cases are essentially the only decidable instances
of language inclusion, in terms of restricting the various resources of
timed automata. 
%
Similar conclusions for the universality problem (does a given TA accept 
all timed words) are drawn in~\cite{AbdullaDOQW08}: the one-clock universality
problem is undecidable for TA over infinite words, and decidable for TA
over finite words, but undecidable for both if $\epsilon$-transitions are allowed.
%
%
Model checking and reachability of timed automata with one or two clocks 
are discussed in~\cite{Laroussinie:2004:Concur,2015:RTT:2795670.2795917}.

It is a fact that reducing resources, like the number of clocks, may
simplify some problems, but allowing $\epsilon$-transitions, even with
few resources, makes the problems as hard as in the general case.
%
%
%
%
A score of papers~\cite{Berard:1996:PNA:646511.759229,opac-b1041550,Diekert:1997:RET:646512.695332,Berard:1998:CEP:2379236.2379238}
investigated the expressiveness of timed automata augmented with 
$\epsilon$-transitions, and proved the following results:
\begin{enumerate}
\item
The class of timed languages recognized by timed automata with 
$\epsilon$-transitions is more robust and expressive than those without them.
\item
A timed automaton with $\epsilon$-transitions that do not reset clocks
can be transformed into an equivalent one without $\epsilon$-transitions
(equivalent means with the same timed language).
\item
A (non-Zenonian) timed automaton such that no $\epsilon$-transitions that 
reset clocks lie on a direct cycle can be transformed into an equivalent 
one without $\epsilon$-transitions.
\item
There is a timed automaton, with an $\epsilon$-transition which resets
clocks on a cycle, which is not equivalent to any timed automaton without
$\epsilon$-transitions.
\end{enumerate}
More undecidability questions for timed automata with $\epsilon$-transitions
were answered in~\cite{Bouyer:2009:URT:2362701.2362702}, e.g.:
given a timed automaton with $\epsilon$-transitions, it is undecidable 
to determine if there exists an equivalent timed 
automaton without $\epsilon$-transitions.
%

%
The problem of removing $\epsilon$-transitions got a new twist 
in~\cite{DBLP:conf/formats/DimaL09}, where it was shown that if one allows 
periodic clock constraints and periodic resets (updates), then we can remove 
$\epsilon$-transitions from a timed automaton; moreover, the authors
proved that periodic updates are necessary, defining a language that cannot
be accepted by any timed automaton with periodic constraints and transitions
which reset clocks to 0 and no $\epsilon$-transitions.
In conclusion, timed automata are a rich model with and without 
$\epsilon$-transitions, therefore in general their decision problems
are undecidable or very difficult also for restricted versions, even
more so if $\epsilon$-transitions are admitted.

An interesting restricted model are Real-Time Automata (RTA)
introduced by C. Dima~\cite{DBLP:journals/jalc/Dima01} in 2001:
they are finite automata with a labeling function (from states to an alphabet)
and a time labeling function (from states to rational intervals) which 
together define the label of a state.
RTA work over signals that are functions with finitely many discontinuities
from non-negative rational intervals $[0,e)$ (with $e > 0$) to an alphabet,
so that the domain of a signal is partitioned into finitely many intervals 
where the signal is constant.
A run is associated with a signal iff there is a sequence of partitioning points
consistent with the state labels (stuttering, i.e., repetition of signal values
is allowed); signals associated with an accepting
run are the timed language associated to an RTA.
The author states in~\cite{DBLP:journals/jalc/Dima01} that RTA can be 
viewed as a class of state-labeled timed automata over timed words 
(instead than signals) with a single clock which is reset at every transition 
(stuttering being reduced to $\epsilon$-transitions).
Moreover, it is claimed that RTA are the largest 
timed extension of finite automata whose emptiness and universality problems 
are decidable, $\epsilon$-transitions can be removed, 
there is a determinization construction, are closed under complementation, 
and a version of Kleene theorem holds.

More complex classes of timed automata have been studied, in which
the interplay between variants of the basic constituents defining them
yields interesting combinations of expressivity and computability.

Event-Clock Automata~\cite{Alur-tcs1999} (ECTA) are a determinizable robust 
subclass of timed automata.
Event-clock automata are characterized with respect to timed automata
by the fact that explicit resets of clocks are replaced by a predefined
association with the input symbols such that for each input $x \in \Sigma$:
a global recorder clock records the time elapsed since the last occurrence
of $x$ and a global predictor clock measures the time required for the next
occurrence of $x$ (clock valuations are determined only by the input timed
words).
They are closed under Boolean operations (TA are not closed under complement)
and language inclusion is PSPACE-complete for them (it is undecidable for TA).
It is mentioned in~\cite{DBLP:journals/jalc/Dima01} that RTA are incomparable
with ECTA, which are the largest known determinizable subclass of timed 
automata. since RTA may accept languages that ECTA cannot.

Timed Automata with Non-Instantaneous Actions~\cite{Barbuti:2001:TAN:1220035.1220037} are such that an action can take some time to be completed;
they are more expressive that timed automata and less expressive
than timed automata with $\epsilon$-transitions.
%
%
Updatable Timed Automata were introduced in~\cite{Bouyer-tcs2004}
as an extension to update the clocks in a more elaborate way than simply 
resetting them to 0; their emptiness problem is undecidable, but
there are interesting decidable subclasses. 
Any updatable automaton belonging to some decidable subclass can be effectively transformed into an equivalent timed automaton without updates, but with $\epsilon$-transitions. 

%
%
%
%
%

%
A complete taxonomy of timed automata is presented in~\cite{Fontana:2014:MTA:2578702.2518102}, and issues of undecidability are discussed in depth in~\cite{Miller:2000:DCR:646880.710453}.
%
Properties of timed automata are contrasted in~\cite{journals/ita/BrihayeBR10} 
with those of a special class of hybrid automata with severe restrictions 
on the discrete transitions: hybrid systems with strong resets, which
have the property that all the continuous 
variables are non-deterministically reset after each discrete transition,
(differently from timed automata, where flow rates are constant, and it is not
compulsory to reset variables on each discrete transition).
%
Connections between timed automata and timed discrete-event models
are explored in~\cite{stavros-fmats-2013}.
%

The trade-off in preferring TA vs.TFSMs depends also on the specific problem 
at hand.
For instance, TA and TFSMs are used when deriving tests for discrete event 
systems.
However, methods for direct derivation of complete test suites over TA 
return infinite test suites~\cite{tretmns-phdthesis}.
Therefore, to derive complete finite test suites with a guaranteed fault 
coverage, a TA is usually converted to an FSM and FSM-based test derivation 
is then used (see~\cite{Springintveld2001,En-Nouaary:2002:TWT:630831.631295}).
Therefore, TFSMs may be preferred over TA and other models when the derivation
of complete tests is required (as done in~\cite{ElFakih-scp2014} for TFSMs
with timed guards), even though the test suites so obtained are rather long.
We mention also that the FSM abstraction introduced in this paper was used
in~\cite{ElFakih-ictss2018}, to derive complete finite test suites for TFSMs
with both timeouts and timed guards.
Since FSMs are used for testing, state distinguishability, and state 
identification problems of hardware and software designs
(see~\cite{testing_fsms-lee_yannakakis-1996,testing_softare_fsms-chow2002,kohavi-book2009}), 
TFSMs may be applied to the timed versions of these problems, instead than
using TA.